\documentclass[11pt]{amsart}
\usepackage{amsmath,amsfonts,amssymb,amsthm,a4wide}
\usepackage[utf8]{inputenc}
\usepackage{tikz}
\usetikzlibrary{decorations.pathreplacing}
\usepackage{enumerate}
\usepackage{xcolor}

\usepackage{todonotes}

\DeclareMathOperator{\bark}{bar}
\DeclareMathOperator{\rep}{rep}

\DeclareMathOperator{\Fac}{Fac}
\DeclareMathOperator{\phifac}{\varphi-\!\Fac}
\theoremstyle{plain}
\newtheorem{theorem}{Theorem}
\newtheorem{lemma}[theorem]{Lemma}
\newtheorem{corollary}[theorem]{Corollary}
\newtheorem{proposition}[theorem]{Proposition}
\theoremstyle{definition}
\newtheorem{definition}[theorem]{Definition}
\newtheorem{remark}[theorem]{Remark}
\newtheorem{example}[theorem]{Example}

\title{Computing the $k$-binomial complexity of the Thue--Morse word}
\author{Marie Lejeune}
\thanks{The first author is supported by a FNRS fellowship}
\author{Julien Leroy}
\author{Michel Rigo}
\email{\{M.Lejeune,J.Leroy,M.Rigo\}@uliege.be}
\address{University of Li\`ege, 
Dept. of Mathematics, 
All\'ee de la d\'ecouverte 12 (B37), 
B-4000 Li\`ege, 
Belgium}

\def\Cut{{\rm Cut}}
\def\bt{\mathbf{t}}

\begin{document}

\begin{abstract}
    Two words are $k$-binomially equivalent whenever they share the same subwords, i.e., subsequences, of length at most $k$ with the same multiplicities. This is a refinement of both abelian equivalence and the Simon congruence. The $k$-binomial complexity of an infinite word $\mathbf{x}$ maps the integer $n$ to the number of classes in the quotient, by this $k$-binomial equivalence relation, of the set of factors of length $n$ occurring in $\mathbf{x}$. This complexity measure has not been investigated very much. In this paper, we characterize the $k$-binomial complexity of the Thue--Morse word. The result is striking, compared to more familiar complexity functions. Although the Thue--Morse word is aperiodic, its $k$-binomial complexity eventually takes only two values. In this paper, we first obtain general results about the number of occurrences of subwords appearing in iterates of the form $\Psi^\ell(w)$ for an arbitrary morphism $\Psi$. We also thoroughly describe the factors of the Thue--Morse word by introducing a relevant new equivalence relation.
\end{abstract}

\maketitle

\section{Introduction}

The Thue--Morse word $\mathbf{t}=011010011001\cdots$ is ubiquitous in combinatorics on words \cite{AS99,Pytheas}. It is an archetypal example of a $2$-automatic sequence: it is the fixed point of the morphism $0\mapsto 01$, $1\mapsto 10$. See, for instance, \cite{AS}. Its most prominent property is that it avoids overlaps, i.e., it does not contain any factors of the form $auaua$ where $u$ is a word and $a$ a symbol. Consequently it also avoids cubes of the form $uuu$ and is aperiodic. The Thue--Morse word appears in many problems with a number-theoretic flavor, to cite a few: the Prouhet--Tarry--Escott problem for partitioning integers, transcendence of real numbers, duplication of the sine,\ldots \cite{AB07,AF08,Dek77,LM}. Let us also mention a sentence from the review of \cite{padic}: ``The nice combinatorial properties of its subword structure have inspired a number of papers'' and Ochsenschl\"ager \cite{Och} was the first to consider the subwords of its prefixes.

\medskip
Various measures of complexity of infinite words have been considered in the literature. In terms of descriptional complexity (i.e., here we are not concerned with algorithms generating infinite words), the most usual one is the {\em factor complexity} that one can, for instance, relate to the topological entropy of a symbolic dynamical system. The factor complexity of an infinite word $\mathbf{x}$ simply counts the number $p_\mathbf{x}(n)=\#\Fac_n(\mathbf{x})$ of factors of length $n$ occurring in $\mathbf{x}$. One can also consider other measures such as arithmetical complexity related to Van der Waerden's theorem \cite{Avg}, abelian complexity introduced by Erd\"os in the sixties (he raised the question whether abelian squares can be avoided by an infinite word over an alphabet of size 4) or, recently $k$-abelian complexity \cite{KS1}. In an attempt to generalize Parikh's theorem on context-free languages, $k$-abelian complexity counts the number of equivalence classes partitioning the set of factors of length $n$ for the so-called $k$-abelian equivalence. Two finite words $u$ and $v$ are {\em $k$-abelian equivalent} if $|u|_x=|v|_x$, for all words of length at most $k$, and where $|u|_x$ denotes the number of occurrences of $x$ as a factor of $u$.
\medskip

The celebrated theorem of Morse--Hedlund characterizes ultimately periodic words in terms of a bounded factor complexity function; for a reference, see \cite{AS} or \cite[Section 4.3]{CANT}. Hence, aperiodic words with the lowest factor complexity are exactly the Sturmian words characterized by $p_\mathbf{x}(n)=n+1$. It is also a well-known result of Cobham that a $k$-automatic sequence has factor complexity in $\mathcal{O}(n)$. The factor complexity of the Thue--Morse word is in $\Theta(n)$ and is recalled in Proposition~\ref{pro:ptm}. 
\medskip

For many complexity measures, Sturmian words have the lowest complexity among aperiodic words, and variations of Morse--Hedlund theorem notably exist for $k$-abelian complexity \cite{KS2}. However, the arithmetical complexity of Sturmian words is in $\mathcal{O}(n^3)$; see \cite{CasFrid}.
\medskip

Binomial coefficients of words have been extensively studied \cite{Lot}: $\binom{u}{x}$ denotes the number of occurrences of $x$ as a subword, i.e., a subsequence, of $u$. They have been successfully used in several applications: $p$-adic topology \cite{padic}, non-commutative extension of Mahler's theorem on interpolation series \cite{Pin}, formal language theory \cite{Karandikar}, Parikh matrices, and a generalization of Sierpi\'nski's triangle \cite{LRS}.
\medskip

Binomial complexity of infinite words has been recently investigated \cite{Rao,RigoSalimov:2015}. The definition is parallel to that of $k$-abelian complexity. Two finite words $u$ and $v$ are {\em $k$-binomially equivalent} if $\binom{u}{x}=\binom{v}{x}$, for all words of length at most $k$. This relation is a refinement of abelian equivalence and Simon's congruence. We thus take the quotient of the set of factors of length $n$ by this new equivalence relation. For all $k\ge 2$, Sturmian words have $k$-binomial complexity that is the same as their factor complexity. However, the Thue--Morse word has bounded $k$-binomial complexity \cite{RigoSalimov:2015}. So we have a striking difference with the usual complexity measures. This phenomenon therefore has to be closely investigated. In this paper, we compute the exact value of the $k$-binomial complexity $b_{\mathbf{t},k}(n)$ of the Thue--Morse word $\mathbf{t}$. To achieve this goal, we first obtain general results computing the number of occurrences of a subword in the (iterated) image by a morphism. This discussion is not restricted to the Thue--Morse morphism.
\medskip

This paper is organized as follows. In Section~\ref{sec:basics}, we recall basic results about binomial coefficients, binomial equivalence and the Thue--Morse word.  In Section~\ref{sec:occ}, we give an expression to compute the coefficient $\binom{\Psi(w)}{u}$ for an arbitrary morphism $\Psi$ in terms of binomial coefficients for the preimage $w$. To that end, we study factorizations of $u$ of the form $u=x\Psi(u')y$. In particular, we are able to express the difference $\binom{\Psi(w)}{u}-\binom{\Psi(w')}{u}$ as a linear combination of the form
$$\sum_x m(x) \left[\binom{w}{x}-\binom{w'}{x}\right],$$
where the sum is ranging over words $x$ shorter than $u$, and we are able to precisely describe the integer coefficients $m(x)$. These coefficients are studied in detail in Section~\ref{sec:mult} where we prove results about $k$-binomially (non)-equivalent factors of the Thue--Morse word of the form $\varphi^k(a)$. In particular, we recover the result of Ochsenschl\"ager about prefixes of the Thue--Morse word \cite{Och}. Indeed, to prove that two words $u,v$ are not $k$-binomially equivalent, it is enough to show that the difference $\binom{u}{x}-\binom{v}{x}$ is non-zero for some word $x$ of length $k$.

In the second part of this paper, we specifically study the $k$-binomial complexity of the Thue--Morse word. For $k=1$, the abelian complexity of $\mathbf{t}$ is well-known and takes only the values $2$ and $3$. The case $k=2$ is treated in Section~\ref{sec:k2}. In the last three sections, we consider the general case $k\ge 3$. The precise statement of our main result is given in Theorem~\ref{the:main}.
The principal tool to get our result is a new equivalence relation discussed in Section~\ref{sec:types}. This relation is based on particular factorizations of factors occurring in the Thue--Morse word. Many authors have been interested in the so-called desubstitution \cite{Grei,mos1,mos2}.

\section{Basics}\label{sec:basics}
Let $A=\{0,1\}$. Let $\varphi:A^*\to A^*$ be the classical Thue--Morse morphism defined by 
$\varphi(0)=01$ and $\varphi(1)=10$. The {\em complement} of a word $u\in A^*$ is the image of $u$ under the involutive morphism mapping $0$ to $1$ and $1$ to $0$. It is denoted by $\overline{u}$. The length of the word $u$ is denoted by $|u|$.
\smallskip

\subsection{Binomial coefficients and binomial equivalence}

The binomial coefficient $\binom{u}{v}$ of two finite words $u$ and $v$ is
the number of times $v$ occurs as a subsequence of $u$ (meaning as a
``scattered'' subword). As an example, we consider two particular words over $\{0,1\}$ and 
$$\binom{101001}{101}=6\ .$$
Indeed, if we index the letters of the first word $u_1u_2\cdots u_6=101001$, we have 
$$u_1u_2u_3=u_1u_2u_6=u_1u_4u_6=u_1u_5u_6=u_3u_4u_6=u_3u_5u_6=101\ .$$
Observe that this concept is a natural generalization of the binomial coefficients of integers. For a one-letter alphabet $\{a\}$, we have
$$\binom{a^m}{a^n}=\binom{m}{n},\quad \forall\, m,n\in\mathbb{N}$$
where $a^m$ denotes the concatenation of $m$ $a$'s. For more on these binomial coefficients, see, for instance, \cite[Chap.~6]{Lot}. In particular, $\binom{u}{\varepsilon}=1$. In this paper, a {\em factor} of a word is made of consecutive letters. However this is not necessarily the case for a {\em subword} of a word.

\begin{definition}[Binomial equivalence]
  Let $k\in\mathbb{N}$ and $u,v$ be two words over $A$. We let $A^{\le k}$ denote the set of words of length at most $k$ over $A$. We say that $u$ and $v$ are {\em $k$-binomially equivalent} if 
  $$\binom{u}{x}=\binom{v}{x},\ \forall x\in A^{\le k}\ .$$
  We simply write $u\sim_k v$ if $u$ and $v$ are $k$-binomially equivalent. The word $u$ is obtained as a permutation of the letters in $v$ if and only if $u\sim_1 v$. In that case, we say that $u$ and $v$ are {\em abelian equivalent} and we write instead $u\sim_{\mathsf{ab}}v$. Note that, for all $k\ge 1$, if $u\sim_{k+1}v$, then $u\sim_k v$.
\end{definition}

\begin{example}\label{exa1}
    The four words $0101110$, $0110101$, $1001101$ and
    $1010011$ are $2$-binomially equivalent. Let $u$ be any of these four words. We have
$$\binom{u}{0}=3,\ \binom{u}{1}=4,\ \binom{u}{00}=3,\ \binom{u}{01}=7,\ \binom{u}{10}=5,\ \binom{u}{11}=6\ .$$
For instance, the word $0001111$ is abelian equivalent to $0101110$ but these two words are not $2$-binomially equivalent. To see this, simply compute the number of occurrences of the subword $10$.
\end{example}

Many classical questions in combinatorics on words can be considered in this binomial context \cite{Rigobook,survey}. Avoiding binomial squares and cubes is considered in \cite{Rao}. The problem of testing whether or not two words are $k$-binomially equivalent is discussed in \cite{testing}. In particular, one can introduce the $k$-binomial complexity function. 

\begin{definition}[Binomial complexity]
    Let $\mathbf{x}$ be an infinite word. The $k$-binomial complexity function of $\mathbf{x}$ is defined as
$$b_{\mathbf{x},k}:\mathbb{N}\to\mathbb{N},\ n\mapsto \# \left( \Fac_n(\mathbf{x})/\!\sim_k\right)$$
where $\Fac_n(\textbf{x})$ is the set of factors of length $n$ occurring in $\textbf{x}$.
\end{definition}

\subsection{Context of this paper}

The Thue--Morse word denoted by $\mathbf{t}$ is the fixed point starting with $0$ of the morphism $\varphi$. In \cite[Thm.~13]{RigoSalimov:2015}, it is shown that $\mathbf{t}$ has a bounded $k$-binomial complexity. Actually, this behavior occurs for all morphisms where images of letters are permutations of the same word.
\begin{theorem}\cite{RigoSalimov:2015} Let $k\ge 1$. There exists
    $C_k>0$ such that the $k$-binomial complexity of the Thue--Morse
    word satisfies $b_{\mathbf{t},k}(n)\le C_k$ for all $n\ge 0$.
\end{theorem}

In the same paper, the following remark was made, \cite[Rem.~5]{RigoSalimov:2015}. 
\begin{remark}
    By computer experiments, $b_{\mathbf{t},2}(n)$ is equal to $9$
    if $n\equiv 0\pmod{4}$ and to $8$ otherwise, for $10\le n\le
    1000$. Moreover, $b_{\mathbf{t},3}(n)$ is equal to $21$ if
    $n\equiv 0\pmod{8}$ and to $20$ otherwise, for $8\le n\le 500$.
\end{remark}

Our contribution is the exact characterization of  $b_{\mathbf{t},k}(n)$. 

\begin{theorem}\label{the:main}
    Let $k$ be a positive integer. 
For all $n \leq 2^k-1$, we have $$b_{\bt,k}(n)=p_\bt(n).$$    
    For all $n\ge 2^k$, we have
$$b_{\mathbf{t},k}(n)=\left\{\begin{array}{l}
3\cdot 2^k-3,\ \text{if } n\equiv 0 \pmod{2^k};\\
3\cdot 2^k-4,\ \text{otherwise}.
\end{array}\right.$$
\end{theorem}
Observe that $3\cdot 2^k-4$ is exactly the number of words of length $2^k-1$ in $\mathbf{t}$, for $k\neq 2 $. Indeed, the factor complexity of $\mathbf{t}$ is well known \cite[Corollary~4.10.7]{CANT}. 

\begin{proposition}\cite{Brlek}\label{pro:ptm}
The factor complexity $p_\mathbf{t}$ of the Thue--Morse word is given by 
$p_{\mathbf{t}}(0)=1$,
$p_{\mathbf{t}}(1)=2$,
$p_{\mathbf{t}}(2)=4$ and for $n\ge 3$, 
$$p_{\mathbf{t}}(n)=\left\{
    \begin{array}{cl}
      4n-2\cdot 2^m-4,& \text{ if }2\cdot 2^m<n\le 3\cdot 2^m;\\
      2n+4\cdot 2^m-2,& \text{ if }3\cdot 2^m<n\le 4\cdot 2^m.\\
    \end{array}\right.$$
\end{proposition}
There are $2$ factors of length $1=2^1-1$ and $6$ factors of length $3=2^2-1$. The number of factors of $\mathbf{t}$ of length $2^k-1$ for $k\ge 3$ is given by $
2(2^k -1)+4 \cdot 2^{k-2} -2 =
3\cdot 2^k-4$, 
$$(p_{\mathbf{t}}(2^k-1))_{k\ge 0}=1,2,6, 20, 44, 92, 188, 380, 764, 1532,\ldots$$
which is exactly one of two values stated in our main result, Theorem~\ref{the:main}.

\subsection{Basic facts about binomial coefficients}
We collect some standard results about binomial coefficients.
\begin{lemma}\label{lem:binomial}
Let $u,v$ be two words and let $a,a'$ be two letters. Then we have
    $$\binom{au}{a'v}=\binom{u}{a'v}+\delta_{a,a'}\binom{u}{v}$$
and 
    $$\binom{ua}{va'}=\binom{u}{va'}+\delta_{a,a'}\binom{u}{v}\ .$$
\end{lemma}

\begin{lemma}\label{lem:binomial2}
Let $s,t,w$ be three words over $A$. Then we have
    $$\binom{sw}{t}=\sum_{\substack{uv=t\\ u,v\in A^*}} \binom{s}{u}\binom{w}{v}\ .$$
\end{lemma}

As a consequence of this lemma, the $k$-binomial equivalence is a congruence. Assume that $u_1\sim_k u_2$ and $v_1\sim_k v_2$, then $u_1v_1\sim_k u_2v_2$.

\begin{lemma}[Cancellation property]\label{lem:cancel}
Let $u,v,w$ be three words. We have
    $$v\sim_k w \Leftrightarrow u\, v\sim_k u\, w
\quad\text{ and }\quad 
    v\sim_k w \Leftrightarrow v\, u\sim_k w\, u\ .$$
\end{lemma}

\begin{proof}
   Since $\sim_k$ is a congruence, we only have to prove that the condition is sufficient. Assume that $v\nsim_k w$. There exists a shortest word $t$, of length at most $k$, such that 
$$\binom{v}{t}\neq\binom{w}{t}\ .$$
We compute
\begin{equation}
    \label{eq:splituvt}
    \binom{uv}{t}=
\sum_{\substack{rs=t\\ r,s\in A^*}} \binom{u}{r}\binom{v}{s}
=\binom{u}{t}+\sum_{\substack{rs=t\\ r,s\in A^+}} \binom{u}{r}\binom{v}{s}+\binom{v}{t}\ .
\end{equation}
In the above formula, $\binom{v}{s}=\binom{w}{s}$ for all $s$ shorter than $t$. Hence, we get exactly the same decomposition for $\binom{uw}{t}$ except for the last term. Thus,
  $$\binom{uv}{t}-\binom{uw}{t}=\binom{v}{t}-\binom{w}{t}\neq 0\ .$$
This means that $uv\nsim_k uw$. Proceed similarly for the second equivalence or observe that $$\binom{\widetilde{x}}{\widetilde{y}}=\binom{x}{y}$$ where $\widetilde{x}$ is the reversal of $x$.
\end{proof}

\begin{lemma}\label{lem:noneq}
    Let $u,v,u',v'$ be four words such that $u\sim_{k-1}u'$ but $u\nsim_{k}u'$ and $v\sim_{k}v'$, then $uv\nsim_k u'v'$. 
\end{lemma}

\begin{proof}
    There exists a word of length $k$ such that $\binom{u}{t}\neq\binom{u'}{t}$ but equality holds for all words shorter than $t$.  One may apply exactly the same reasoning as in~\eqref{eq:splituvt}.
\end{proof}

%=============================================================================

\section{Occurrences of subwords in images by $\varphi$}\label{sec:occ}

The aim of this section is to obtain an expression for coefficients of the form $\binom{\varphi(w)}{u}$. Even though we are mainly interested in the Thue--Morse word, our observations can be applied to any non-erasing morphism as summarized by Theorem~\ref{the:general_statement}. 

A {\em multiset} is just a set where elements can be repeated with a (finite) integer multiplicity. If $x$ belongs to a multiset $M$, its multiplicity is denoted by $m_M(x)$ or simply $m(x)$. If $x\not\in M$, then $m_M(x)=0$. If we enumerate the elements of a multiset, we adopt the convention to write multiplicities with indices. The {\em multiset sum} $M\uplus N$ of two multisets $M,N$ is the union of the two multisets and the multiplicity of an element is equal to the sum of the respective multiplicities, i.e., for $x\in M\cup N$, $m_{M\uplus N}(x)=m_M(x)+m_N(x)$. For instance, $\{1_2,3_1,4_3\}\uplus \{1_3,2_1,3_4\}=\{1_5,2_1,3_5,4_3\}$.  

Let us start with an introductory example. We hope that this example will forge the intuition of the reader about the general scheme.
\begin{example}\label{exa:intro}
We want to compute
$$\binom{\varphi(0110001)}{u}\quad\text{ with }u=01011\ .$$
The word $w=\varphi(0110001)$ belongs to $\{01,10\}^*$. It can be factorized with consecutive blocks $b_1b_2\cdots b_7$ of length~$2$. To count the number of occurrences of the subword $u$ in the image by $\varphi$ of a word, several cases need to be taken into account:
\begin{itemize}
  \item the five symbols of $u$ appear in pairwise distinct $2$-blocks of $w$ (each $2$-block contains both $0$ and $1$ exactly once), and there are 
$$\binom{|w|/2}{|u|}=\binom{7}{5}$$
such choices;
\item the prefix $01$ of $u$ is one of the $2$-blocks $b_i$ of $w$ and the last three symbols of $u$ appear in subsequent pairwise distinct $2$-blocks $b_j$, $j>i$. Since $\varphi(0)=01$, we have to count the number of occurrences of the subword $0z$, for all words $z$ of length $3$, in the preimage of $w$. There are 
$$\sum_{z\in A^3} \binom{0110001}{0z}=\binom{6}{3}+1$$
such choices;
\item the first symbol of $u$ appear in a $2$-block, the first occurrence of $10$ in $u$ appears as a whole $2$-block and the last two symbols $11$ occur in subsequent pairwise distinct $2$-blocks. There are 
$$\sum_{x\in A}\sum_{z\in A^2} \binom{0110001}{x1z}
=\binom{5}{2}+\binom{2}{1}\binom{4}{2}$$
such choices;
\item the first two symbols of $u$ appear in distinct $2$-blocks, the second occurrence of $01$ in $u$ appears as a whole $2$-block and the last symbol $1$ occurs in a subsequent $2$-block. There are 
$$\sum_{x\in A^2}\sum_{z\in A} \binom{0110001}{x0z}
=\binom{3}{2}\binom{3}{1}+\binom{4}{2}\binom{2}{1}+\binom{5}{2}$$
such choices;
\item finally, $u$ can appear as two $2$-blocks $\varphi(0)$ followed by a single letter occurring in a subsequent $2$-block. There are 
$$\sum_{z\in A} \binom{0110001}{00z}
=\binom{3}{1}+\binom{2}{1}\binom{2}{1}+\binom{3}{1}$$
such choices.
\end{itemize}
\end{example}
The general scheme behind this computation is expressed by Theorem~\ref{the:phifac} given below. The reader can already feel that we need to take into account particular factorizations of $u$ with respect to occurrences of a factor $\varphi(0)$ or $\varphi(1)$. The five cases discussed in Example~\ref{exa:intro} correspond to the following factorizations of $u$:
$$01011,\ \varphi(0)011,\ 0\varphi(1)11,\ 01\varphi(0)1,\ \varphi(0)\varphi(0)1$$ 
We thus introduce the notion of a $\varphi$-factorization.

\begin{definition}[$\varphi$-factorization]
    If a word $u\in A^*$ contains a factor $01$ or $10$, then it can be factorized as
\begin{equation}
    \label{eq:factorization}
    u=w_0\, \varphi(a_1)\, w_1 \cdots w_{k-1} \, \varphi(a_k)\, w_k
\end{equation}
for some $k\ge 1$, $a_1,\ldots,a_k\in A$ and $w_0,\ldots,w_k\in A^*$ (some of these words are possibly empty). We call this factorization, a {\em $\varphi$-factorization} of $u$. It is coded by the $k$-tuple of positions where the $\varphi(a_i)$'s occurs: 
$$\kappa=(|w_0|, |w_0\varphi(a_1)w_1|,|w_0\varphi(a_1)w_1\varphi(a_2)w_2|,\ldots,|w_0\varphi(a_1)w_1\varphi(a_2)w_2\cdots w_{k-1}|)\ .$$
The set of all the $\varphi$-factorizations of $u$ is denoted by $\phifac(u)$. 
\end{definition}

Since $|\varphi(a)|=2$, for all $a\in A$, observe that if $(i_1,\ldots,i_k)$ codes a $\varphi$-factorization, then $i_{j+1}-i_j\ge 2$ for all $j$. Note that $u$ starts with a prefix $01$ or $10$ if and only if there are $\varphi$-factorizations of $u$ coded by tuples starting with $0$.

\begin{example}
Consider the word $010110$. It has $9$ $\varphi$-factorizations. These factorizations and the coding tuples are depicted in Figure~\ref{fig:phifact}.

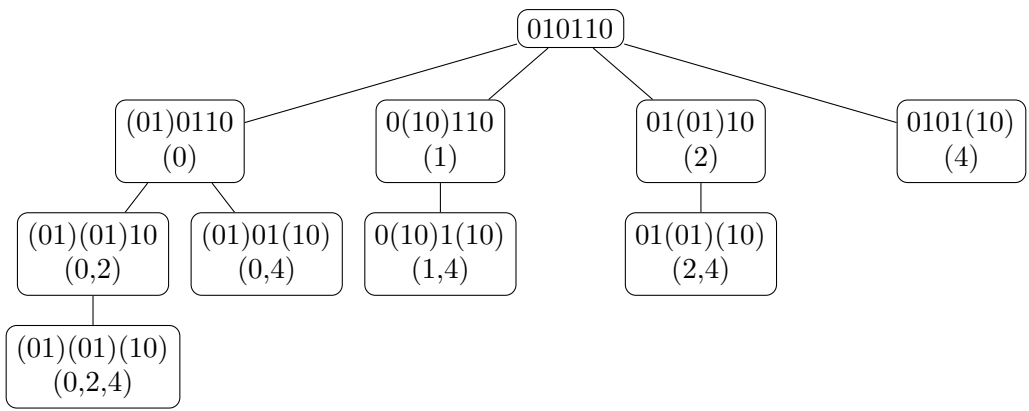
\begin{figure}[h!tb]
    \centering
\begin{tikzpicture}[
  every node/.style = {shape=rectangle, rounded corners, draw, align=center},
level 1/.style={sibling distance=9em},
level 2/.style={sibling distance=6em}
]
  \node {010110}
    child { node {(01)0110\\ (0)} 
     child {node {(01)(01)10\\ (0,2)} 
      child {node {(01)(01)(10)\\ (0,2,4)}}
           }
     child {node {(01)01(10)\\ (0,4)} }
          }
    child { node {0(10)110\\ (1)}
     child { node {0(10)1(10)\\ (1,4)}} 
          }
    child { node {01(01)10\\ (2)} 
     child { node {01(01)(10)\\ (2,4)}} }
    child { node {0101(10)\\ (4)} }
;
\end{tikzpicture}
\caption{The tree of $\varphi$-factorizations of $010110$.}
    \label{fig:phifact}
\end{figure}
\end{example}

We define a map $f$ from $A^*$ to the set of finite multisets of words over $A^*$. This map is defined as follows. 
\begin{definition}
If $u\in 0^*\cup 1^*$, then $f(u)=\emptyset$ (the meaning for this choice will be clear with Theorem~\ref{the:phifac}). If $u$ is not of this form, it contains a factor $01$ or $10$. With every $\varphi$-factorization $\kappa\in\phifac(u)$ of $u$ of the form~\eqref{eq:factorization}
$$u=w_0\, \varphi(a_1)\, w_1 \cdots w_{k-1} \, \varphi(a_k)\, w_k$$
for some $k\ge 1$, $a_1,\ldots,a_k\in A$ and $w_0,\ldots,w_k\in A^*$,  we define the language 
$$\mathcal{L}(u,\kappa):=A^{|w_0|}\, a_1\, A^{|w_1|} \cdots A^{|w_{k-1}|} a_k A^{|w_k|}$$
of words of length $|u|-k$ (there are $2^{|u|-2k}$ of these words\footnote{We have all the words of length $|u|-k$ where in $k$ positions the occurring symbol is given.}). Such a language is considered as a multiset whose elements have multiplicities equal to $1$. Now, $f(u)$ is defined as the multiset sum (i.e., we sum the multiplicities) of the above languages for all $\varphi$-factorizations of $u$, i.e., 
$$f(u):=\biguplus_{\kappa\in \phifac(u)} \mathcal{L}(u,\kappa)\ .$$
\end{definition}

Note that for $u\not\in 0^*\cup 1^*$, $f(u)$ only contains words of length less than $|u|$. In particular, since there always exist $\varphi$-factorizations coded by a $1$-tuple, then $f(u)$ contains words of length $|u|-1$. If $|u|\ge 2$, the languages of the form $A^{|w_0|}a_1A^{|w_1|}$ associated with this $\varphi$-factorization coded by a $1$-tuple contains at least one word not in $0^*\cup 1^*$ because $|w_0|+|w_1|=|u|-1\ge 1$.

\begin{example}\label{exa:fu}
Consider the word $u=01011$. It has four $\varphi$-factorizations of the form~\eqref{eq:factorization}
$$(01)011,\ 
0(10)11,\ 
01(01)1,
(01)(01)1\ .$$
The first three are coded respectively by the $1$-tuples $(0)$, $(1)$ and $(2)$. The last one is coded by $(0,2)$. The corresponding four languages are
$$0\, A^3=\{0 000, 0 001, 0 010, 0 011, 0 100, 0 101, 0 110, 0111\}, $$
$$A\, 1\, A^2=\{0100, 0101, 0110, 0111, 1100, 1101, 1110, 1111\}, $$
$$A^2\, 0\, A=\{0000, 0001, 0100, 0101, 1000, 1001, 1100, 1101\}, $$
$$00\, A=\{000,001\}\ .$$
Consequently, $f(u)$ is the multiset (multiplicities are indicated as indices)
\begin{eqnarray*}
  \{000_1,001_1,0000_2,0001_2,0010_1,0011_1,0100_3,0101_3,0110_2,0111_2,\\
    1000_1,1001_1,1100_2,1101_2,1110_1,1111_1\}.
\end{eqnarray*}
\end{example}

\begin{definition}
    Now that $f$ is defined over $A^*$, we can extend it to any finite multiset $M$ of words over $A$. It is the multiset sum of the $f(v)$'s, for all $v\in M$, repeated with their multiplicities.
\end{definition}

\begin{example}
Continuing Example~\ref{exa:fu} with $u=01011$, we get 
$$f^2(u)=\{
00_{4}, 01_{2}, 10_{2}, 000_{20}, 001_{16}, 010_{28}, 011_{24}, 100_{12}, 
101_{8}, 110_{20}, 111_{16} \}$$
and
$$f^3(u)=\{
0_2, 1_2, 00_{76}, 01_{100}, 10_{44}, 11_{68}\}\ .$$
Observe that if we apply $f$ an extra time, $f^4(u)=\{0_{100},1_{44}\}$ and for all $n\ge 5$, $f^n(u)=\emptyset$.
\end{example}

\begin{remark}\label{rem:length}
    The observation made in the previous example is general. If $u$ does not belong to $0^*\cup 1^*$, then $f^{|u|-2}(u)$ contains only elements in $\{0,1,00,01,10,11 \}$ and $f^{|u|-1}(u)$ contains only elements in $\{0,1\}$. For $n\ge |u|$, $f^n(u)$ is empty.
\end{remark}

Recall that $f(u)$ is a multiset. Hence $m_{f(u)}(v)$ denotes the multiplicity of $v$ as element of $f(u)$.

\begin{theorem}\label{the:phifac} With the above notation, for all words $u,w$, we have
    $$\binom{\varphi(w)}{u}= \binom{|w|}{|u|} + \sum_{\substack{\kappa\in \phifac(u) \\v\in \mathcal{L}(u,\kappa)}} \binom{w}{v}=\binom{|w|}{|u|} + \sum_{v\in f(u)}m_{f(u)}(v)\, \binom{w}{v}
\ .$$
\end{theorem}

\begin{proof}
   The reader should probably reconsider the introductory Example~\ref{exa:intro}. The result directly follows from the definitions of $\varphi$-factorization and $\mathcal{L}(u,\kappa)$. If $u$ belongs to $0^*\cup 1^*$, then there is no $\varphi$-factorization of $u$. The stated formula is reduced to the first term: we have to pick the symbols of $u$ in pairwise distinct $2$-blocks of $\varphi(w)$. Otherwise, we also have to consider all the cases corresponding to the $\varphi$-factorizations of $u$ where some factors of $u$ are realized by $\varphi(0)$ or $\varphi(1)$.
\end{proof}

\begin{corollary}\label{cor:phik} 
Let $k\ge 1$. For all words $u,v$, we have
    $$u\sim_k v\Rightarrow \varphi(u)\sim_{k+1}\varphi(v)\ .$$
\end{corollary}

\begin{proof}
    Let $t$ be a word of length at most $k+1$. From Theorem~\ref{the:phifac}, we have
$$\binom{\varphi(u)}{t}=\binom{|u|}{|t|}+\sum_{s\in f(t)} m_{f(t)}(s)\, \binom{u}{s}.$$
For all $s\in f(t)$, we have $|s|\le k$ and thus $\binom{u}{s}=\binom{v}{s}$. The conclusion follows since $|u|=|v|$.
\end{proof}

Theorem~\ref{the:phifac} can be extended to iterates of $\varphi$. If we apply it twice, we get
\begin{eqnarray*}
\binom{\varphi^2(w)}{u}&=& 
\binom{|\varphi(w)|}{|u|} + \sum_{v\in f(u)}m_{f(u)}(v)\binom{\varphi(w)}{v}\\
&=&\binom{|\varphi(w)|}{|u|} + \sum_{v\in f(u)}m_{f(u)}(v)
\left[ \binom{|w|}{|v|} + \sum_{z\in f(v)}m_{f(v)}(z)\binom{w}{z}
\right]\\
&=&\binom{|\varphi(w)|}{|u|} + \sum_{v\in f(u)}m_{f(u)}(v)
 \binom{|w|}{|v|} + \sum_{x\in f^2(u)}m_{f^2(u)}(x)\binom{w}{x}.
\end{eqnarray*}
The last equality comes from the fact that
$$\sum_{v\in f(u)} \sum_{z\in f(v)} m_{f(u)}(v)\, m_{f(v)}(z)
=\sum_{x\in f^2(u)}m_{f^2(u)}(x)\ .$$
Indeed, $z$ appears $m_{f(v)}(z)$ times in the multiset $f(v)$ and $v$ itself appears $m_{f(u)}(v)$ times in $f(u)$. Thus $z$ appears $m_{f(u)}(v)\, m_{f(v)}(z)$ in $f^2(u)$.\\
We set $f^0(u)=\{u\}$ (where $u$ has multiplicity one).

\begin{corollary}\label{cor:phifac}
With the above notation, for $\ell\ge 1$ and all words $u,w$, we have
$$\binom{\varphi^\ell(w)}{u}=
\sum_{i=0}^{\ell-1} \sum_{v\in f^i(u)}m_{f^i(u)}(v)
 \binom{|\varphi^{\ell-i-1}(w)|}{|v|} +
 \sum_{x\in f^\ell(u)}m_{f^\ell(u)}(x)\binom{w}{x}\ .$$
\end{corollary}

\begin{proof}
    Proceed by induction on $\ell$.
\end{proof}

When proving that two words $x,y$ are not $k$-binomially equivalent, it is convenient to find a word $u$ of length $k$ such that the difference $\binom{x}{u}-\binom{y}{u}$ is non-zero. It is therefore interesting to make the following observation.

\begin{remark}\label{rem:difference}
Since $|\varphi^i(w)|=2^i|w|$, note that the first of the two terms in the above Corollary only depends on $|w|$ and $u$. Otherwise stated, if $w,w'$ are two words of the same length, then
$$\binom{\varphi^\ell(w)}{u}-\binom{\varphi^\ell(w')}{u}=
\sum_{x\in f^\ell(u)}m_{f^\ell(u)}(x)\left[\binom{w}{x}-\binom{w'}{x}\right]\ .$$
\end{remark} 

The reader should be convinced that the following general statement holds.
\begin{theorem}\label{the:general_statement}
    Let $\Psi:A^*\to B^*$ be a non-erasing morphism and $u\in B^+,w\in A^+$ be two words.
$$\binom{\Psi(w)}{u}=\sum_{k=1}^{|u|}
\sum_{\substack{u_1,\ldots,u_k\in A^+\\ u=u_1\cdots u_k}}
\sum_{a_1,\ldots,a_k\in A} 
\binom{\Psi(a_1)}{u_1}\cdots \binom{\Psi(a_k)}{u_k}\binom{w}{a_1\cdots a_k}.$$
\end{theorem}

The word $u$ occurs as a subword of $\Psi(w)$ if and only if there exists $k\ge 1$ such that $u$ can be factorized into $u_1\cdots u_k$ where, for all $i$, $u_i$ is a non-empty subword occurring in $\Psi(a_i)$ for some letter $a_i$ and such that $a_1\cdots a_k$ is a subword of $w$.

\begin{corollary}
With the above notation, if $w$ and $w'$ are two words of the same length, we have
$$\binom{\Psi(w)}{u}-\binom{\Psi(w')}{u}
=\sum_{k=1}^{|u|}
\sum_{\substack{u_1,\ldots,u_k\in A^+\\ u=u_1\cdots u_k}}
\sum_{a_1,\ldots,a_k\in A} 
\binom{\Psi(a_1)}{u_1}\cdots \binom{\Psi(a_k)}{u_k}
\left[
\binom{w}{a_1\cdots a_k}-\binom{w'}{a_1\cdots a_k}
\right]
.$$
\end{corollary}

\section{About multiplicities}\label{sec:mult}

In this section we give more insight about multiplicities of the form $m_{f^\ell(u)}(x)$ appearing in Corollary~\ref{cor:phifac}. This will permit us to prove results about $k$-binomially (non)-equivalent factors of the Thue--Morse word of the form $\varphi^k(a)$.

\begin{lemma}\label{lem:same_mult} Let $w$ be a word. 
    Let $M$ be a (finite) multiset of words such that $u$ belongs to $M$ if and only if its complement $\overline{u}$ belongs to $M$ with the same multiplicity. For all $i\ge 0$, the multiplicity of $w$ in $f^i(M)$ is equal to the one of $\overline{w}$.
\end{lemma}

\begin{proof}
Let $u$ be a word in $M$. Because of the special form of the morphism $\varphi$, we deduce that the set of tuples coding the $\varphi$-factorizations of $u$ is equal to the set of tuples coding the $\varphi$-factorizations of $\overline{u}$. Moreover, a word $v$ belongs to $\mathcal{L}(u,\kappa)$ if and only if $\overline{v}$ belongs to $\mathcal{L}(\overline{u},\kappa)$. Indeed, these two languages are respectively of the form 
$$ A^{|w_0|}\, a_1\, A^{|w_1|} \cdots A^{|w_{k-1}|} a_k A^{|w_k|}
\quad\text{ and }\quad
 A^{|w_0|}\, \overline{a_1}\, A^{|w_1|} \cdots A^{|w_{k-1}|} \overline{a_k} A^{|w_k|}\ .$$
Think about Example~\ref{exa:fu} and consider the word $\overline{u}=10100$, 
$$u=(01)011,\ 
0(10)11,\ 
01(01)1,
(01)(01)1$$
and
$$\overline{u}=(10)100,\ 
1(01)00,\ 
10(10)0,
(10)(10)0\ .$$
For instance, the third $\varphi$-factorization gives, respectively, the languages
$$A^2\, 0\, A\quad\text{ and }\quad A^2\, 1\, A\ .$$
Let $w$ be a word over $A$. 
Since $u$ and $\overline{u}$ have the same multiplicity, the total number of times $w$ occurs in the $m(u)$ copies of $\mathcal{L}(u, \kappa)$ is equal to the number of times $\overline{w}$ occurs in the copies of $\mathcal{L}(\overline{u}, \kappa)$. This observation holds true for every $\varphi$-factorization. Consequently $f(M)$ has the same property as $M$: words and their complement appear with the same multiplicity in $f(M)$. We can thus iterate the construction and the argument. 
\end{proof}

\begin{example}
    Consider the multiset $M=\{
011_2, 100_2, 0110_1, 1001_1
\}$. In $f(M)$ the words $00$ and $11$ have multiplicity $2$, $01$ and $10$ have multiplicity $3$ and all words of length $3$ appear twice. Then
$$f^2(M)=\{0_3,1_3,00_{8},11_{8},01_8,10_8\}\ .$$
\end{example}

\begin{proposition}\label{pro:diffpower}
   For all $n\ge 1$, the multiplicity of $01$ (resp., of $00$) in the multiset $f^n(01^{n+1})=f^{n-1}(0A^n)$ is larger than the one of $10$ (resp., of $11$). More precisely, these multiplicities in the multiset $f^n(01^{n+1})$ satisfy
\[
	m(01) - m(10) = m(00) - m(11) = 1 \cdot 2 \cdot 4 \cdot 8 \cdots 2^{n-1}=2^{n(n-1)/2}.
\]
\end{proposition}

\begin{proof}
    We proceed by induction on $n$. For $n=1$, $f^0(0A)=\{00,01\}$ and the result is obvious. Let $n\ge 2$. Assume that the result holds for all $j<n$. We consider $f^{n-1}(0A^{n})$. 

 Note that $0A^{n}$ is the disjoint union of $\{00u\mid u\in A^{n-1}\}$ and $\{01\overline{u}\mid u\in A^{n-1}\}$. These two sets are in one-to-one correspondence with the map $0w\mapsto 0\overline{w}$. Since we proceed by induction, let us start by applying $f$ once. We will apply $f^{n-2}$ later on.

Let $u\in A^{n-1}$. First observe that there is a one-to-one correspondence between the set of $\varphi$-factorizations of $00u$ and the set of $\varphi$-factorizations of $01\overline{u}$ coded by tuples whose first element is at least equal to $1$. In this case, we have exactly the same tuples. For instance, consider the word $u=1011$.
$$\begin{array}{c|c|c}
00u=001011 & \kappa & 01\overline{u}=010100 \\
\hline
0(01)011 & (1) & 0(10)100 \\
00(10)11 & (2) & 01(01)00 \\
001(01)1 & (3) & 010(10)0 \\
0(01)(01)1 & (1,3)& 0(10)(10)0 \\
\end{array}$$
For each such $\varphi$-factorization $\kappa$ of $00u$, we have a language of the form 
$$\mathcal{L}(00u,\kappa)=A^{i_0} a_1 A^{i_1} \cdots A^{i_{k-1}} a_k A^{i_k}$$
with $i_0\ge 1$, $k\ge 1$, $a_1,\ldots,a_k\in A$, $i_1,\ldots,i_k\ge 0$ and the corresponding $\varphi$-factorization of $01\overline{u}$ gives the language

$$\mathcal{L}(01\overline{u},\kappa)=A^{i_0} \overline{a_1} A^{i_1} \cdots A^{i_{k-1}} \overline{a_k} A^{i_k}\ .$$

Observe that the union of these two languages satisfies the assumption of the previous lemma. 
Thus, applying iteratively $f$ to the words belonging to these languages will eventually provide words $01$ and $10$ (resp., $00$ and $11$) with the same multiplicity. We stress the fact that in the above $\varphi$-factorizations, $i_0$ is non-zero. 

We still have to consider the $\varphi$-factorizations of $01\overline{u}$ coded with tuples starting with $0$ (these are the only remaining ones). With the running example $u=1011$, we have the extra three $\varphi$-factorizations:
$$(01)0100,\ (01)(01)00,\ (01)0(10)0\ .$$ 

Let us consider $\varphi$-factorizations of $01\overline{u}$  coded by a $k$-tuple starting with $0$ and for $k\ge 2$. Thus the resulting languages are made of words of length at most $|u|+2-k=n+1-k\le n-1$. By Remark~\ref{rem:length}, applying $f^{n-2}$ to words of such lengths will only provide words in $\{0,1\}$. Hence, they do not provide any copy of $00, 01, 10$ or $11$. 

We finally have the $\varphi$-factorization of $01\overline{u}$ coded by the $1$-tuple $(0)$. The corresponding language $\mathcal{L}(01\overline{u},(0))$ is $0A^{|u|}$. Recall that $u$ (and thus $\overline{u}$) is ranging over $A^{n-1}$. Thus there are $2^{n-1}$ copies of this language. By induction hypothesis, the difference of multiplicities for $01$ and $10$ (resp., $00$ and $11$) for $f^{n-2}(0A^{n-1})$ is $1 \cdot 2\cdots 2^{n-2}$. Multiplying the latter number by the number of copies provides us with the result.
\end{proof}

Table~\ref{tab:mult0110} provides the computed multiplicities of $01$ and $10$ in  $f^{n-1}(0A^n)$ for the first few values of $n$. We also indicate the corresponding differences given in the previous proposition. 
\begin{table}
$$\begin{array}{c|rrr}
& m(01) & m(10) & m(01) - m(10)\\
\hline
1&1 & 0 & 1=2^0\\
2&3 & 1 & 2=2^1\\ 
3&28 & 20 & 8=2^3\\ 
4&800 & 736 & 64=2^6\\ 
5&61952 & 60928 & 1024=2^{10}\\ 
6&11812864 & 11780096 & 32768=2^{15} \\
7&5285871616 & 5283774464 & 2097152=2^{21}\\
\end{array}$$    
    \caption{Multiplicities in $f^{n-1}(0A^n)$.}
    \label{tab:mult0110}
\end{table}

\begin{proposition}\label{pro:diffpower2}
   For all $n\ge 1$, the multiplicity of $0$ in the multiset $f^n(01^{n+1})=f^{n-1}(0A^{n})$ is larger than or equal to the multiplicity of $1$. 
\end{proposition}

\begin{proof}
For $n=1,2$, there is no $0$ and no $1$ in $f^n(01^{n+1})$. Assume $n\ge 3$.
    The multiplicity of $0$ (resp., $1$) in the multiset $f^n(01^{n+1})$ is equal to the multiplicity of $01$ (resp., $10$) in $f^{n-1}(01^{n+1})=f^{n-2}(0A^n)$. One can thus follow the lines of the proof of Proposition~\ref{pro:diffpower} except for $\varphi$-factorizations starting with $0$ and for $k \geq 2$ where a more careful discussion is needed. 

Consider the $\varphi$-factorizations of the words $01\overline{u}$ of length $n+1$ coded by a $k$-tuple starting with $0$ and for $k\ge 2$. If $k>2$, when applying $f$ once, the resulting languages are made of words of length less than $n-1$ and applying $f^{n-3}$ to these words will provide no $01$ nor $10$. But for $k=2$, applying $f$ once to all such words, we  
get the multiset $0A^{n-2}$ where each word has multiplicity $2^{n-3}$. 
Indeed, the word $\overline{u}$ is ranging over $A^{n-1}$ and we consider $\varphi$-factorizations where one replacement of $01$ or $10$ is made inside $\overline{u}$. For all $i,j\le n-3$, we have
\begin{eqnarray*}
&\#\{x01y\in A^{n-1}\mid |x|=i\}= 
\#\{x01y\in A^{n-1}\mid |x|=j\}\\
=&\#\{x10y\in A^{n-1}\mid |x|=i\}=
\#\{x10y\in A^{n-1}\mid |x|=j\}=2^{n-3}.  
\end{eqnarray*}
To conclude the proof, observe that the difference of multiplicity of $01$ and $10$ in $f^{n-3}(0A^{n-2})$ is obtained from the previous proposition.
\end{proof}

%==========================================================================

\subsection{Some consequences for the factors of Thue--Morse}
We collect some important properties of iterates of $\varphi$ with respect to the $k$-binomial equivalence $\sim_k$. A trace of the first result below can be found in \cite{Och}.

\begin{lemma}[Ochsenschl\"ager]\label{lem:equiv-not} Let $k\ge 1$. We have 
$$\varphi^k(0)\sim_k\varphi^k(1)
\quad\text{ and }\quad 
\varphi^k(0)\nsim_{k+1}\varphi^k(1)\ .$$ 
In particular, if $|u|=|v|$, then $\varphi^k(u)\sim_k\varphi^k(v)$.
\end{lemma}

\begin{proof}
    We have $\varphi(0)\sim_1\varphi(1)$. Thus the first part follows from Corollary~\ref{cor:phik}. 
\smallskip

Let us show that $\varphi^k(0)\nsim_{k+1}\varphi^k(1)$. The case $k=1$ is obvious: $01\nsim_2 10$. Observe that $f(01^k)=0A^{k-1}$ thus $f^{k-1}(01^k)=f^{k-2}(0A^{k-1})$. 
Using Remark~\ref{rem:difference}, we compute 
$$\binom{\varphi^{k}(0)}{01^k}-\binom{\varphi^k(1)}{01^k}$$
and get 
$$\binom{\varphi^{k-1}(\varphi(0))}{01^k}-\binom{\varphi^{k-1}(\varphi(1))}{01^k}
=\sum_{v\in f^{k-1}(01^k)} m_{f^{k-1}(01^k)}(v) \left[\binom{\varphi(0)}{v}-\binom{\varphi(1)}{v}\right].$$
The elements of the multiset $f^{k-1}(01^k)$ belong to $\{0,1,00,01,10,11\}$. The last factor in brackets in the previous sum is non-zero only if $v=01$ or $v=10$. Hence, we get  
$$\binom{\varphi^{k}(0)}{01^k}-\binom{\varphi^k(1)}{01^k}=
m_{f^{k-1}(01^k)}(01) -m_{f^{k-1}(01^k)}(10)=2^{(k-1)(k-2)/2}\ge 1\ .$$
The last equality comes from Proposition~\ref{pro:diffpower}. 
\end{proof}

\begin{lemma}[Transfer lemma]
\label{lem:equiv2} Let $k\ge 1$. Let $u,v,v'$ be three non-empty words such that $|v|=|v'|$. We have 
    $$\varphi^{k-1}(u)\, \varphi^{k}(v)\sim_k \varphi^{k}(v')\, \varphi^{k-1}(u)\ .$$
\end{lemma}

\begin{proof}
    Observe that $u\varphi(v)\sim_1 \varphi(v')u$ because $v$ and $v'$ have the same length. The conclusion follows from Corollary~\ref{cor:phik}: $\varphi^{k-1}(u\, \varphi(v))\sim_k\varphi^{k-1}(\varphi(v')\, u)$.  
\end{proof}

\begin{corollary}\label{cor:permut}
    Let $k\ge 1$ and $n\ge 2$. Let $u_1,\ldots,u_n$ be non-empty words. Let  $i_1,\ldots,i_n$ be integers greater than or equal to $k$, except for one of these being equal to $k-1$ and denoted by $i_r$. For all permutations $\nu$ of $\{1,\ldots,n\}$, we have  
$$\varphi^{i_1}(u_1)\varphi^{i_2}(u_2)\cdots \varphi^{i_n}(u_n)\sim_k
\varphi^{i_{\nu(1)}}(u_{\nu(1)}')\varphi^{i_{\nu(2)}}(u_{\nu(2)}')\cdots \varphi^{i_{\nu(n)}}(u_{\nu(n)}')$$
for all words $u_1',\ldots,u_n'$ where $|u_i|=|u_i'|$, for all $i$, and $u_{i_r}=u_{i_r}'$.
\end{corollary}

\begin{proof}
    It is enough to see that one can permute any two consecutive factors: any permutation is a product of such type of transpositions. This is a direct consequence of the two previous lemmas. 
\end{proof}

%===============================================================================================================

\section{$2$-binomial complexity}\label{sec:k2}
In this section we compute the value of $b_{\mathbf{t},2}(n)$.
First of all, the next proposition ensures us that all the words we will consider in the proof of Theorem~\ref{thm:k2} really appear as factors of $\bt$. The reader familiar with B\"uchi's theorem and the characterization of $k$-automatic sequence in terms of first-order logic can obtain an alternative proof of this result. Basically, one has to check that the four closed formulas (for $a,b\in\{0,1\}$)
$$(\forall m)(\exists i)(\mathbf{t}_i=a\wedge \mathbf{t}_{i+m+1}=b)$$
hold. This can be done automatically using the Walnut package \cite{walnut}. Here, we proceed with a classical proof relying on the definition of the Thue--Morse word in terms of the base-$2$ sum-of-digits function. We let $\rep_2(n)$ denote the base-$2$ expansion of $n$. If $n>0$, we assume that $\rep_2(n)$ starts with a $1$ (i.e., has no leading zeroes).

\begin{proposition}
\label{pro:factorsappear}
Let $k,m \in \mathbb{N}$ and $a,b \in\{0,1\}$. Let $p_u$ be a suffix of $\varphi^{k}(a)$ and $s_u$ be a prefix of $\varphi^{k}(b)$. There exists $z \in \{0,1\}^m$ such that $p_u \varphi^{k}(z) s_u$ is a factor of $\mathbf{t}$.
\end{proposition}
\begin{proof}
Let $a,b \in \{0,1\}, m \in \mathbb{N}$. We will prove that  there exists $z \in \{0,1\}^m$ such that $azb \in \Fac(\mathbf{t})$. Therefore, $\varphi^{k}(a) \varphi^k(z) \varphi^k (b)$ (and thus, $p_u \varphi^k(z) s_u$) is a factor of $\mathbf{t}$. We  set $\ell = |\rep_2(m+1)|$.

For all $n \in \mathbb{N}$, we let $\bt_n$ denote the $(n+1)^{\text{th}}$ letter of $\bt$. From \cite{AS}, we know that 
\begin{align*}
\bt_n = |\rep_2(n)|_1 \bmod 2.
\end{align*}

The idea of the proof is to find $x \in \mathbb{N}$ such that $\bt_x = a$ and $\bt_{x+m+1} =b$. Eight cases have to be considered depending on the parity of $a,b,\ell$. Let us first assume $a=b=0$. 

If $|\rep_2(m+1)|_1 \equiv 0 \pmod 2$, we can take $x = 2^{\ell+1} + 2^{\ell +2}$, so $\bt_x= 0$, 
$\rep_2(x+m+1)=110\cdot \rep_2(m+1)$ where $\cdot$ is just the concatenation product and thus, evaluating the parity of $1$'s in the expansion, we get $\bt_{x+m+1} = 0$. 

Otherwise, $|\rep_2(m+1)|_1 \equiv 1 \pmod 2$ and we can set $x = 2^{\ell -1} + 2^{\ell +1}$. Thus, we have to take care of a carry for $2^{\ell-1}+2^{\ell-1}=2^{\ell}$ and 
$$\rep_2(x+m+1) = 110\cdot 0^j\rep_2(m+1- 2^{\ell-1})$$ 
where $j=\ell-1-|\rep_2(m+1- 2^{\ell-1})|$ and thus, $\bt_{x+m+1}=1$.

The other cases are similarly treated. If $a=0,b=1$, one can verify that
\begin{align*}
x = \left\{\begin{array}{ll}
2^{\ell+1} + 2^{\ell +2} ,& \text{if } |\rep_2(m+1)|_1 \equiv 1 \pmod 2; \\
2^{\ell+1} + 2^{\ell -1}, & \text{otherwise.}
\end{array}
\right.
\end{align*}
is convenient.
Similarly, if $a=1, b=0$, we can set
\begin{align*}
x = \left\{\begin{array}{ll}
2^{\ell+1},  & \text{if } |\rep_2(m+1)|_1 \equiv 1 \pmod 2; \\
2^{\ell -1}, & \text{otherwise.}
\end{array}
\right.
\end{align*}
and finally, if $a=b=1$, then take
\begin{align*}
x = \left\{\begin{array}{ll}
2^{\ell+1},  & \text{if } |\rep_2(m+1)|_1 \equiv 0 \pmod 2; \\
 2^{\ell -1}, & \text{otherwise.}
\end{array}
\right.
\end{align*}
\end{proof}

Using this result, we can compute the values of $b_{\bt,2}$.

\begin{theorem}\cite[Thm. 3.3.6]{Lejeune}\label{thm:k2}
We have $b_{\bt,2}(0)=1$, $b_{\bt,2}(1)=2$, $b_{\bt,2}(2)=4$, $b_{\bt,2}(3)=6$ and for all $n \geq 4$, 
\begin{align*}
b_{\mathbf{t},2}(n) = \left\{
\begin{array}{ll}
9, & \text{if } n \equiv 0 \pmod 4; \\
8, & \text{otherwise.}
\end{array}
\right.
\end{align*}
\end{theorem}

\begin{proof}
Assume $n \geq 4$.
First observe that, for all words $u,v$ of the same length, 
\[
u \sim_2 v \Leftrightarrow \binom{u}{0} = \binom{v}{0} \, \text{and} \, \binom{u}{01} = \binom{v}{01}.
\]
Indeed, this is due to the fact that $\binom{u}{1} = |u| - \binom{u}{0}$, $\binom{u}{aa} = \binom{|u|_a}{2}$ for every $a \in \{0,1\}$ and $\binom{u}{10} = \binom{|u|}{2} - \binom{u}{00} - \binom{u}{01} - \binom{u}{11}$.

We will consider four cases depending on the value of $\lambda \in \{0, 1,2,3 \}$ such that $n \equiv \lambda \pmod 4$. For every one of them, we will compute 
\[
b_{\bt,2}(n) = \# \left\{ \left(\binom{u}{0}, \binom{u}{01} \right) \in \mathbb{N} \times \mathbb{N} : u \in \Fac_n(\textbf{t}) \right\}.
\]

Since $\textbf{t}$ is the fixed point of the morphism $\varphi$, we know that every factor $u$ of length $n$ of $\textbf{t}$ can be written $p_u \varphi^2(z) s_u$ for some $z \in A^*$ and $p_u$ (resp., $s_u$) suffix (resp., prefix) of a word in $\{ \varphi^2(0), \varphi^2(1) \}$. From the previous proposition, we also know that every word of that form occurs at least once in $\bt$.
Moreover, we have $|p_u| + |s_u| \in \{\lambda,\lambda+4\}$ and, as a consequence, $|z| = \left\lfloor \frac{n}{4} \right\rfloor$ = $\frac{n-\lambda}{4}$ or $|z| = \left\lfloor \frac{n}{4} \right\rfloor - 1$.
Set $\ell = \frac{n-\lambda}{4}$.

Let us first consider the case $\lambda=0$. We have
\begin{align*}
\Fac_n(\bt) = \{ 
&\varphi^2(az), 0 \varphi^2(z) 011, 0 \varphi^2(z) 100, 1 \varphi^2(z) 011, 1 \varphi^2(z) 100,\\ &01 \varphi^2(z) 01, 01 \varphi^2(z) 10, 10 \varphi^2(z) 01,10 \varphi^2(z) 10, \\ &110 \varphi^2(z) 0,110 \varphi^2(z) 1, 001 \varphi^2(z) 0,001 \varphi^2(z) 1 : z \in A^{\ell -1}, a \in A, az \in \Fac(\bt) 
\}
\end{align*}

Let us illustrate the computation of $\left(\binom{u}{0}, \binom{u}{01} \right)$ on $u=0\varphi^2(z)011 \in \Fac_n(\bt)$. Firstly,
$$\binom{u}{0} = \binom{0}{0} + \binom{\varphi^2(z)}{0} + \binom{011}{0}  = 2+2|z| = 2 \ell$$
since $|z| = \ell-1$. 
Similarly, we have
\begin{align*}
\binom{u}{01} &= \binom{0}{01} + \binom{\varphi^2(z)}{01} + \binom{011}{01} + \binom{0}{0} \binom{\varphi^2(z)}{1} + \binom{0}{0} \binom{011}{1} + \binom{\varphi^2(z)}{0} \binom{011}{1} \\
&= \binom{|\varphi(z)|}{2} + \binom{\varphi(z)}{0} + 2 + |\varphi(z)| + 2 + 2|\varphi(z)| \\
&= |z| (2|z|-1) + |z| + 6|z| + 4 = 2 \ell^2 + 2 \ell.
\end{align*}

All the computations are summarized in the table below. We give the form of the factors and respective values for the pairs $\left(\binom{u}{0}, \binom{u}{01} \right)$.
\begin{align*}
\begin{array}{|c|c|c|c|c|c|}
\hline
\text{Case}& \varphi^2(az) &  0 \varphi^2(z) 011 & 1\varphi^2(z) 100 & 0 \varphi^2(z) 100 & 001\varphi^2(z) 0  \\
& 01 \varphi^2(z) 10 & 001 \varphi^2(z) 1& 110 \varphi^2(z) 0 & & \\
& 10 \varphi^2(z) 01 & & & & \\
\hline
\binom{u}{0} & 2\ell & 2\ell & 2\ell & 2\ell+1 & 2\ell+1 \\
\binom{u}{01} & 2\ell^2 & 2\ell^2+2\ell & 2\ell^2-2\ell & 2\ell^2-1 & 2\ell^2 \\
\hline
\hline
\text{Case} &  1 \varphi^2(z) 011  & 110\varphi^2(z) 1 & 01 \varphi^2(z) 01 & 10 \varphi^2(z) 10 &   \\
\hline
\binom{u}{0} & 2\ell-1 & 2\ell-1 & 2\ell & 2\ell & \\
\binom{u}{01} & 2\ell^2 & 2\ell^2 +1 & 2\ell^2+1 & 2\ell^2-1 & \\
\hline
\end{array}
\end{align*}

This is thus clear that if $n \equiv 0 \pmod 4$, we have $b_{\bt,2}(n)  =9$.

Let us now present the results in the case $\lambda=1$. Let $a$ be a letter and $z$ be any word of $A^{\ell-1}$. We obtain the results below, showing that $b_{\bt,2}(n)=8$.
\begin{align*}
\begin{array}{|c|c|c|c|c|}
\hline
\text{Case}& \varphi^2(az) 0  & 0 \varphi^2(az)  & \varphi^2(az) 1 & 1 \varphi^2(az)  \\
& 01 \varphi^2(z) 100  & 110 \varphi^2(z) 01 & 10 \varphi^2(z) 011 & 001 \varphi^2(z) 10 \\
\hline
\binom{u}{0} & 2\ell+1 & 2\ell+1 & 2\ell  & 2\ell \\
\binom{u}{01} & 2\ell^2 & 2\ell^2+2\ell &2\ell^2+2\ell & 2\ell^2 \\
\hline
\hline
\text{Case} & 10\varphi^2(z) 100 & 001 \varphi^2(z) 01 &   01\varphi^2(z) 011 & 110 \varphi^2(z)10  \\
\hline
\binom{u}{0} & 2\ell+1  & 2\ell+1 & 2\ell & 2\ell \\
\binom{u}{01} & 2\ell^2-1  & 2\ell^2+2\ell+1 & 2\ell^2+2\ell +1 & 2\ell^2-1 \\
\hline
\end{array}
\end{align*}

In the case of $\lambda=2$, if $z$ is a word of $A^{\ell-1}$ and if $a$ is a letter, we obtain $b_{\bt,2}(n)=8$ due to the following results:

\begin{align*}
\begin{array}{|c|c|c|c|c|}
\hline
\text{Case}& \varphi^2(az) 01 & \varphi^2(az) 10 & 1 \varphi^2(az) 1 & 0 \varphi^2(az) 0  \\
& 01\varphi^2(az) & 10 \varphi^2(az) & 110 \varphi^2(z) 011& 001 \varphi^2(z) 100 \\
\hline
\binom{u}{0} & 2\ell+1 & 2\ell+1 & 2\ell & 2\ell+2 \\
\binom{u}{01} & 2\ell^2+2\ell+1 &2\ell^2+2\ell &2\ell^2+2\ell & 2\ell^2+2\ell \\
\hline
\hline
\text{Case} & 1\varphi^2(az) 0 & 0\varphi^2(az) 1 & 001 \varphi^2(z) 011 & 110 \varphi^2(z) 100  \\
\hline
\binom{u}{0} & 2\ell+1 & 2\ell+1 & 2\ell+1 & 2\ell+1 \\
\binom{u}{01} & 2\ell^2 & 2\ell^2+4\ell +1 & 2\ell^2+4\ell+2 & 2\ell^2-1 \\
\hline
\end{array}
\end{align*} 

Finally, if $\lambda=3$ and using the same notation, we obtain that $b_{\bt,2}(n)=8$ due to the following computations,

\begin{align*}
\begin{array}{|c|c|c|c|c|}
\hline
\text{Case}& \varphi^2(az) 011 & 110 \varphi^2(az) & \varphi^2(az) 100 & 001 \varphi^2(az)  \\
& 01 \varphi^2(az) 1 & 1 \varphi^2(az) 10 &  10 \varphi^2(az) 0 & 0 \varphi^2(az) 01 \\
\hline
\binom{u}{0} & 2\ell+1& 2\ell+1  & 2\ell+2 & 2\ell+2 \\
\binom{u}{01} & 2\ell^2+4\ell+2 &2\ell^2+2\ell &2\ell^2+2\ell  & 2\ell^2+4\ell+2 \\
\hline
\hline
\text{Case} & 1\varphi^2(az) 01& 10 \varphi^2(az)1  & 0\varphi^2(az) 10 & 01 \varphi^2(az) 0  \\
\hline
\binom{u}{0} & 2\ell+1  & 2\ell+1 & 2\ell+2 & 2\ell+2 \\
\binom{u}{01} & 2\ell^2+2\ell+1 & 2\ell^2+4\ell+1 & 2\ell^2+4\ell +1  & 2\ell^2+2\ell+1 \\
\hline
\end{array}
\end{align*}
which concludes the proof.
\end{proof}

\section{How to cut factors of the Thue--Morse word}

Computing $b_{\mathbf{t},k}(n)$, for all $k\ge 3$, will require much more knowledge about the factors of $\mathbf{t}$. This section is concerned about particular factorizations of factors occurring in $\mathbf{t}$.

Since $\mathbf{t}$ is a fixed point of $\varphi$, it is very often convenient to view $\mathbf{t}$ as a concatenation of blocks belonging to $\{\varphi^k(0),\varphi^k(1)\}$. Hence, we first define a function $\bark_k$ that roughly plays the role of a ruler marking the positions where a new block of length $2^k$ occurs (these positions are called {\em cutting bars of order~$k$}). For all $k \geq 1$, let us consider the function $\bark_k: \mathbb{N} \to \mathbb{N}$ defined by 
\[
	\bark_k(n) = |\varphi^k(\mathbf{t}_{[0,n)})| = n \cdot 2^k,
\]
where $\mathbf{t}_{[0,n)}$ is the prefix of length $n$ of $\mathbf{t}$.

Given a factor $u$ of $\mathbf{t}$, we are interested in the relative positions of $\bark_k(\mathbb{N})$ in $u$: we look at all the occurrences of $u$ in $\mathbf{t}$ and see what configurations can be achieved, that is how an interval $I$ such that $\bt_I = u$ can intersect $\bark_k(\mathbb{N})$.

For instance, for $k=1$, the word $u = 010$ occurs in $\mathbf{t}$ with two different factorizations:
\begin{equation}
	\begin{array}{rcccccccccc}
	\mathbf{t}  &=& \varphi(0) & \varphi(1) & \varphi(1) &\varphi(0)&\varphi(1) &\varphi(0) &\varphi(0) &\varphi(1)& \cdots \\
				&=& 01 &  1 0 & 10  & 01  & 10 & 01 & 01 & 10 & \cdots \\
				&=& \multicolumn{9}{l}{\;\, 01 \; \cdot \;\; 1 \boxed{0\; \cdot\;\;\; 10}{} \; \cdot \;\; 01 \;\; \cdot \;\; 10 \; \cdot \;\; \boxed{01 \; \cdot \;\; 0}1 \; \cdot \;\; \;10 \;\;\;\;\, \cdots}
	\end{array}
\end{equation}
				
The first occurrence of $010$ is obtained as a suffix of $\varphi(11)$ and the second one as a prefix of $\varphi(00)$. The dots represented in the above figure are representing the {\em cutting bars} (of order~$1$) of the substitution. So, we see that for the factor $010$, two kinds of configurations of the cutting bars can be achieved.

\begin{definition}[Cutting set]
\label{def:cuttingsets}
For all $k \geq 1$, we define the set $\Cut_k(u)$ of non-empty sets of relative positions of cutting bars
\[
\Cut_k(u):=\biggl\{ \bigl( [i,i+|u|] \, \cap \, \bark_k(\mathbb{N})\bigr)-i 
	\mid i \in \mathbb{N}, u = \mathbf{t}_{[i,i+|u|)} \biggr\}\, .
\] 
A \textit{cutting set of order} $k$ is an element of $\Cut_k(u)$. Observe that we consider the closed interval $[i,i+|u|]$ because we are also interested in knowing if the end of $u$ coincide with a cutting bar.
\end{definition}
To continue with our example, we have $\Cut_1(010) = \{\{1,3\},\{0,2\}\}$, meaning that $u$ contains two cutting bars and the first one is situated before or after the first letter.
We also represent this by
\[
	{\rm Cut_1}(010) = \{0 \cdot 10 \ \cdot, \cdot \ 01 \cdot 0\}.
\]

\begin{remark}
Let $u$ be a factor of $\bt$. Observe that, for all $\ell \geq 1$, $\Cut_\ell(u) \neq \emptyset$. 
It results from the following three observations. 
Obviously, $\bark_{k}(\mathbb{N})\subset \bark_{k-1}(\mathbb{N})$ and thus if $\Cut_k(u)$ is non-empty, then the same holds for $\Cut_{k-1}(u)$.

Next notice that if $\Cut_k(u)$ contains a singleton, then $\Cut_{k+1} (u)$ contains a singleton. Indeed, we can write $u=u_1 u_2$ with $u_1$ a suffix of $\varphi^k(a)$, $u_2$ a prefix of $\varphi^k(b)$. Thus $u_1$ is a suffix of $\varphi^{k+1}(\overline{a})$ and $u_2$ is a prefix of $\varphi^{k+1}(b)$.

Finally, there exists a unique $k$ such that $2^{k-1} \leq |u| \leq 2^k-1$. There also exists $i$ such that $u= \bt_{[i,i+|u|)}$. 
Simply notice that either $[i,i+|u|] \cap \bark_k(\mathbb{N})$ is a singleton or,  
$[i,i+|u|] \cap \bark_{k-1}(\mathbb{N})$ is a singleton. The conclusion follows.
\end{remark}

Observe that for any word $u$ and any set $C \in \Cut_k(u)$, there is a unique integer $r \in \{0,1,\dots,2^{k}-1\}$ such that $C \subset 2^k \mathbb{N}+r$.

\begin{lemma}
\label{lemma:unique_desubstitution}
Let $k$ be a positive integer, $u$ be a factor of $\bt$ and $C = \{i_1 < i_2 < \cdots < i_n\}$ be a set in $\Cut_k(u)$.
There is a unique factor $v$ of $\bt$ of length $n-1$ such that
$u = p \varphi^k(v) s$, with $|p|=i_1$.
Furthermore, if $i_1 > 0$ (resp., $i_n < |u|$), there is a unique letter $a$ such that $p$ (resp., $s$) is a proper suffix (resp., prefix) of $\varphi^k(a)$.
\end{lemma}
\begin{proof}
Since $u_{[i_1,i_n)}$ belongs to $\varphi^k(A^*)$, the uniqueness of $v$ follows from the injectivity of~$\varphi$.
For the uniqueness of $a$, this follows from the fact that $\varphi^k(0)$ and $\varphi^k(1)$ do not have any (non-empty) common prefix or suffix.
\end{proof}

\begin{definition}[Factorization of order~$k$]
Given a factor $u$ of $\mathbf{t}$ of length at least $2^k-1$, there always exists a set $C$ in $\Cut_k(u)$.
By Lemma~\ref{lemma:unique_desubstitution}, we can associate with $C$ a unique pair 
\[
(p,s) \in A^* \times A^*
\]
and a unique triple 
\[
(a,v,b) \in (A \cup \{\varepsilon\}) \times A^* \times (A \cup \{\varepsilon\})
\]
such that $u = p \varphi^k(v) s$, where 
either $a=p=\varepsilon$ 
(resp., $b=s=\varepsilon$), 
or $a \neq \varepsilon$ and $p$ is a proper suffix of $\varphi^k(a)$ 
(resp., $b \neq \varepsilon$ and $s$ is a proper prefix of $\varphi^k(b)$).
In particular, we have $a=p=\varepsilon$ exactly when $\min(C) = 0$ and $b=s=\varepsilon$ exactly when $\max(C)=|u|$.
The triple $(a,v,b)$ is called the {\em desubstitution of $u$ associated with $C$} and the pair $(p,s)$ is called the {\em factorization of $u$ associated with $C$}.
If $C \in \Cut_k(u)$, then $(a,v,b)$ and $(p,s)$ are respectively desubstitutions and factorizations of {\em order} $k$.
\end{definition}

If $u$ is a factor of $\bt$ of length at least $4$, then $\Cut_1(u)$ contains a single set. 
Indeed, the factors of length 4 of $\bt$ are
\[
	\{0010, 0011, 0100, 0101, 0110, 1001, 1010, 1011, 1101, 1100\}.
\]
If the word $00$ or $11$ occurs in $u$, then $u$ necessarily has a cutting bar between the two occurrences of $0$ or of $1$ and this cutting bar determines all the others, forcing the set $\Cut_1(u)$ to be a singleton.
If $00$ and $11$ do not occur in $v$, then $v = 0101$ or $v=1010$, those  cases being symmetric.
If $v = 0101$, then the potential cutting bars are 
\[
	\cdot \ 01 \cdot 01 \ \cdot 
	\qquad \text{or} \qquad
	0 \cdot 10 \cdot 1. 
\]
However, the second case implies that the factor $v$ occurs in $\bt$ as a factor of $\varphi(111)$. 
As $111$ is not a factor of $\bt$, this shows that $\Cut_1(v) = \{\cdot \ 01 \cdot 01 \ \cdot\}$.

In the following statement, taking $k \geq 3$ ensures that we consider long enough words to have a unique set in $\Cut_1(u)$. 

\begin{lemma}
\label{lemma: bijection cut k cut k-1}
Let $k \geq 3$ be an integer, $u$ be a factor of $\bt$ of length at least $2^k-1$.
Let $(a,v,b)$ be the desubstitution of $u$ associated with the unique set in $\Cut_1(u)$ and let us write $u = p\varphi(v)s$, with $p$ suffix of $\varphi(a)$ and $s$ prefix of $\varphi(b)$.
Let finally $C$ be a set in $\Cut_k(u)$. 
\begin{enumerate}
\item
If $\min C < 2^k-1$ and $|u| - \max C < 2^k-1$,
then the set $C' = (C+|p|)/2$ belongs to $\Cut_{k-1}(avb)$;

\item
If $\min C = 2^k-1$ and $|u| - \max C < 2^k-1$,
then the set $C' = \{0\} \cup (C+|p|)/2$ belongs to $\Cut_{k-1}(avb)$;

\item
If $\min C < 2^k-1$ and $|u| - \max C = 2^k-1$,
then the set $C' = \{|avb|\} \cup (C+|p|)/2$ belongs to $\Cut_{k-1}(avb)$;

\item
If $\min C = 2^k-1$ and $|u| - \max C = 2^k-1$,
then the set $C' = \{0,|avb|\} \cup (C+|p|)/2$ belongs to $\Cut_{k-1}(avb)$.
\end{enumerate}
Moreover, the application from $\Cut_k(u)$ to $\Cut_{k-1}(avb)$ that maps $C$ to $C'$ is a bijection.
\end{lemma}

\begin{proof}
We consider the desubstitution $(a_0,a_1a_2\cdots a_t,a_{t+1})$ associated with $C$.
By definition, there is a unique $r \in \{0,1,\dots, 2^{k}-1\}$ such that 
$C = \{r,r+2^k,r+2 \cdot 2^k, \dots, r+t\cdot 2^k\}$ and we have 
\[
	u = \alpha \varphi^k(a_1 \cdots a_t) \beta,
\]
with $\alpha$ the suffix of length $r$ of $\varphi^k(a_0)$ and $\beta$ the prefix of length $|u|-r-t\cdot 2^k$ of $\varphi^k(a_{t+1})$.
There exist words $u_1,u_2,\dots,u_m, u_1',u_2',\dots,u_n'$ in $\varphi(A)$ such that 
\[
	\alpha = p u_1 u_2 \cdots u_m
	\quad \text{and} \quad
	\beta = u_1' u_2' \cdots u_n' s.	
\]
Let $v_1,\dots,v_m,v_1',\dots,v_n' \in A$ such that $u_i = \varphi(v_i)$ and $u_i' = \varphi(v_i')$ for all $i$.
We have 
\begin{equation}
\label{eq:facto k-1 avb}
	avb =  a v_1 \cdots v_m \varphi^{k-1}(a_1 \cdots a_t) v_1' \cdots v_n' b
\end{equation} 
and $a v_1 \cdots v_m$ is a suffix of $\varphi^{k-1}(a_0)$ and $v_1' \cdots v_n' b$ is a prefix of $\varphi^{k-1}(a_{t+1})$.

If $|\alpha|,|\beta|< 2^k-1$, then 
\[
	|a v_1 \cdots v_m| = \lceil r/2 \rceil< 2^{k-1}
	\quad \text{and} \quad
	|v_1' \cdots v_nb|< 2^{k-1}.	
\]
Therefore, the set $C' \in \Cut_{k-1}(avb)$ associated with the factorization~\eqref{eq:facto k-1 avb} is 
\[
	C' = \{\lceil r/2 \rceil, \lceil r/2 \rceil + 2^{k-1} , \dots, \lceil r/2 \rceil + t \cdot 2^{k-1}\}.
\]
For the other cases, if for instance $|\alpha| = 2^k-1$, then $|a v_1 \cdots v_m|= 2^{k-1}$, which explains why we add $0$ in $C'$.

Let us show that the correspondence between $C$ and $C'$ is bijective.
It is trivially surjective.
If there is some other cutting set $D = \{r',r'+2^k,\dots \}$ in $\Cut_k(u)$, then $|r-r'| \geq 2$ since both $C$ and $D$ must be included in the unique set of $\Cut_1(u)$.
This shows that the associated sets $C',D' \in \Cut_{k-1}(avb)$ are different.
\end{proof}

The substitution $\varphi$ being primitive and $\mathbf{t}$ being aperiodic, Moss\'e's recognizability theorem ensures that the substitution $\varphi^k$ is {\em bilaterally recognizable}~\cite{mos1,mos2} for all $k \geq 1$, i.e., any sufficiently long factor $u$ of $\mathbf{t}$ can be uniquely desubstituted by $\varphi^k$ (up to a prefix and a suffix of bounded length).
In the case of the Thue--Morse substitution, we can make this result more precise. Similar results are considered in \cite{Grei} where the term (maximal extensible) reading frames is used.

\begin{lemma}
\label{lemma:cut singleton}
Let $k$ be a positive integer.
If $u$ is a factor of $\bt$ of length $|u| > 3 \cdot 2^{k-1}$, then $\Cut_k(u)$ is a singleton.
\end{lemma}

\begin{proof}
First observe that given a word $u$ and a prefix $v$ of $u$, a set of cutting bars for $v$ can be extended in a unique way into a set of cutting bars for $u$.
More precisely, if $v$ is a prefix of $u$ and if $C$ belongs to $\Cut_k(v)$, there is a unique set $C'$ such that $C' \in \Cut_k(u)$ and $C \subset C'$.
It is thus enough to prove the result for words of length exactly $3\cdot 2^{k-1}+1$.

We proceed by induction on $k$.
The case $k = 1$ has already been considered before Lemma~\ref{lemma: bijection cut k cut k-1}.
Let us now assume that the result is true for all $\ell \leq k$ and let us prove it for $k+1$.
If $|u|=3\cdot 2^{k}+1$, then by the induction hypothesis, $\Cut_1(u)$ is a singleton and, using Lemma~\ref{lemma:unique_desubstitution}, there is a unique factor $v$ of $\bt$ such that 
\begin{enumerate}
\item
$u$ is a factor of $\varphi(v)$;
\item
$u$ is not a factor of $\varphi(v')$ for any proper factor $v'$ of $v$.
\end{enumerate}
Since $u$ is a factor of $\varphi(v)$, we have $|u| \leq 2 |v|$ and thus, since $|v|$ is an integer, $|v| > 3\cdot 2^{k-1}$.
Using again the induction hypothesis and Lemma~\ref{lemma:unique_desubstitution}, there is a unique factor $w$ of $\bt$ such that
\begin{enumerate}
\item
$v$ is a factor of $\varphi^k(w)$;
\item
$v$ is not a factor of $\varphi^k(w')$ for any proper factor $w'$ of $w$.
\end{enumerate}
This word $w$ is thus the unique factor of $\bt$ such that
\begin{enumerate}
\item
$u$ is a factor of $\varphi^{k+1}(w)$;
\item
$u$ is not a factor of $\varphi^{k+1}(w')$ for any proper factor $w'$ of $w$.
\end{enumerate}
This shows that $\Cut_{k+1}(u)$ is a singleton.
\end{proof}

\begin{lemma}
\label{lem:relationCuts}
Let $k \geq 3$ be an integer and $u$ be a factor of $\bt$ of length $2^k-1 \leq |u| \leq 3 \cdot 2^{k-1}$.
Then $\Cut_k(u)$ is a not a singleton if and only if $u$ is a factor of $\varphi^{k-1}(010)$ or of $\varphi^{k-1}(101)$, in which case $\Cut_k(u) = \{C_1,C_2\}$ and $|\min C_1 - \min C_2|=2^{k-1}$.
In this case, let  $(p_1,s_1)$, $(p_2,s_2)$ be the two factorizations of order $k$ respectively associated with $C_1, C_2 \in \Cut_k(u)$. Without loss of generality, assume that $|p_1|<|p_2|$. Then, there exists $a \in A$ such that either
 \[|p_1|+|s_1| = |p_2|+|s_2|
 \text{ and }  
 	(p_2, \varphi^{k-1}(a) s_2) = (p_1 \varphi^{k-1}(a), s_1)
 \]

or, 
 \[ ||p_1|+|s_1| - (|p_2|+|s_2|)| = 2^k \text{ and }   
 	(p_2, s_2) = (p_1 \varphi^{k-1}(\bar{a}), \varphi^{k-1}(a) s_1).
 \]

\end{lemma}

\begin{proof}
The case $k=3$ can be checked by hand.
Assume that the result holds for $k \geq 3$ and let us prove for $k+1$.
Let $(a,v,b)$ be the desubstitution of $u$ associated with the unique set in $\Cut_1(u)$. 
By Lemma~\ref{lemma: bijection cut k cut k-1}, we have $\# \Cut_{k+1}(u) = \# \Cut_{k}(avb)$ and if $\Cut_{k}(avb) = \{C_1',C_2'\}$ and $\Cut_{k+1}(u) = \{C_1,C_2\}$, then $|\min C_1 - \min C_2| = 2 |\min C_1' - \min C_2'|$. 
Furthermore, $u$ is a factor of $\varphi^{k}(010)$ (resp., of $\varphi^{k}(101)$) if and only if $avb$ is a factor of $\varphi^{k-1}(010)$ (resp., of $\varphi^{k-1}(101)$).

For the last part of the proof, first assume that $u$ is a factor of $\varphi^{k-1}(a \bar{a} a)$, but not a prefix nor a suffix.
Since $|u| \geq 2^k-1$, we have $u = u' \varphi^{k-1}(\bar{a}) u''$, with $u'$ and $u''$ respectively suffix and prefix of $\varphi^{k-1}(a)$, $|u'|,|u''|<2^{k-1}$.
Therefore, $u$ admits the two cuttings sets
\[
	u' \cdot \varphi^{k-1}(\bar{a}) u''
	\quad \text{and} \quad
	u' \varphi^{k-1}(\bar{a}) \cdot u''.	
\]
The associated factorizations are
\[
	(u' , \varphi^{k-1}(\bar{a}) u'')
	\quad \text{and} \quad
	(u' \varphi^{k-1}(\bar{a}), u'')	
\]
so we are in the first situation.

Assume now that $u$ is a prefix of $\varphi^{k-1}(a \bar{a} a)$; the case where $u$ is a suffix is similar.
Two cases can occur: either $\varphi^{k-1}(a \bar{a})$ is a prefix of $u$, or $u$ is a proper prefix of $\varphi^{k-1}(a \bar{a})$.
If $\varphi^{k-1}(a \bar{a})$ is a prefix of $u$, then $u = \varphi^{k-1}(a \bar{a}) u'$ for some prefix $u'$ of $\varphi^{k-1}(a)$.
If $|u'|<2^{k-1}$, the two cutting sets of order $k$ of $u$ are
\[
	\cdot \, \varphi^k(a) \cdot u'
	\quad \text{and} \quad
	\varphi^{k-1}(a) \cdot \varphi^{k-1}(\bar{a}) u'
\]
and the associated factorizations are respectively
\[
	(\varepsilon,u')
	\quad \text{and} \quad
	(\varphi^{k-1}(a) , \varphi^{k-1}(\bar{a}) u').
\]
We are thus in the second situation. Else, $u' = \varphi^{k-1}(a)$, the two cutting sets of order $k$ of $u$ are
\[
	\cdot \, \varphi^k(a) \cdot u'
	\quad \text{and} \quad
	\varphi^{k-1}(a) \cdot \varphi^{k-1}(\bar{a}) u' \, \cdot
\]
and the associated factorizations are respectively
\[
	(\varepsilon,u')
	\quad \text{and} \quad
	(\varphi^{k-1}(a) , \varepsilon).
\]
We are in the first situation.

If $u$ is a proper prefix of $\varphi^{k-1}(a \bar{a})$, then $u = \varphi^{k-1}(a) u'$ where $u'$ is the prefix of length $2^{k-1}-1$ of $\varphi^{k-1}(\bar{a})$ (because $|u| \geq 2^k-1$).
The two cutting sets of order $k$ of $u$ are
\[
	\cdot \, \varphi^{k-1}(a) u'
	\quad \text{and} \quad
	\varphi^{k-1}(a) \cdot u'
\]
and the associated factorizations are respectively
\[
	(\varepsilon,\varphi^{k-1}(a) u')
	\quad \text{and} \quad
	(\varphi^{k-1}(a) , u').
\]
We are thus in the first situation.
\end{proof}

\section{Types associated with a factor}\label{sec:types}

\begin{remark}
    All the following constructions rely on Lemma~\ref{lem:relationCuts}. Thus, in the remaining of this paper, we will always assume that $k \geq 3$.
\end{remark}

Lemma~\ref{lem:relationCuts} ensures us that whenever a word has two cutting sets,  
then their associated factorizations are strongly related.
We will now show that whenever two factors $u, v$ of the same length of $\bt$ admits factorizations of order $k$ that are similarly related, then these two words are $k$-binomially equivalent.
   
To this aim, we introduce an equivalence relation $\equiv_k$ on the set of pairs $(x,y) \in A^{< 2^k} \times A^{< 2^k}$. 
The core result of this section is given by Theorem~\ref{thm:typediff} stating that two words are $k$-binomially equivalent if and only if their factorizations of order~$k$ are equivalent for this new relation $\equiv_k$. 
So, the computation of $b_{\mathbf{t},k}(n)$ amounts to determining the number of equivalence classes for $\equiv_k$ among the factorizations of order~$k$ for words in $\Fac_n(\mathbf{t})$. 

\begin{definition}
\label{def:type}
Two pairs $(p_1,s_1)$ and $(p_2,s_2)$ of $A^{< 2^k} \times A^{< 2^k}$ are equivalent for $\equiv_k$ whenever there exists $a \in A$ such that one of the following situations occurs:
\begin{enumerate}
\item
$|p_1|+|s_1| = |p_2|+|s_2|$ and
\begin{enumerate}
\item 
$(p_1,s_1) = (p_2,s_2)$;
\item
$(p_1, \varphi^{k-1}(a) s_1) = (p_2 \varphi^{k-1}(a), s_2)$;
\item
$(p_2, \varphi^{k-1}(a) s_2) = (p_1 \varphi^{k-1}(a), s_1)$; 
\item
$(p_1,s_1) = (s_2,p_2) = (\varphi^{k-1}(a),\varphi^{k-1}(\bar{a}))$; 
\end{enumerate}
\item
$\bigl| |p_1|+|s_1| - (|p_2|+|s_2|) \bigr| = 2^k$ and
\begin{enumerate}
\item
$(p_1,s_1) = (p_2\varphi^{k-1}(a),\varphi^{k-1}(\bar{a})s_2)$; 
\item
$(p_2,s_2) = (p_1\varphi^{k-1}(a),\varphi^{k-1}(\bar{a})s_1)$.
\end{enumerate}
\end{enumerate}
\end{definition}
\begin{remark}\label{rem:lgequivk}
    Note that if $(p_1,s_1)\equiv_k(p_2,s_2)$, then either $|p_1|=|p_2|$ or, $\left| |p_1|-|p_2| \right|=2^{k-1}$. So $(p_1,s_1)\equiv_k(p_2,s_2)$ implies that $|p_1|\equiv |p_2| \pmod{2^{k-1}}$. 
\end{remark}
\begin{example}
\label{ex:type}
Let us consider $k=3$ and 
\begin{eqnarray*}
	u =	01 0110 01101001 10010110 100 
	  	&=& 01\varphi^{2}(0) \varphi^{3} (01) 100,	\\
	v = 01 10010110 10010110 0110 100
		&=& 01 \varphi^{3}(11) \varphi^2(0)100.			
\end{eqnarray*}
From Lemma~\ref{lemma:cut singleton}, they admit a unique factorization of order 3 that are respectively 
\[
	(p_u,s_u) = (01 \varphi^2(0), 100) 
	\quad \text{and} \quad 
	(p_v,s_v) = (01, \varphi^2(0) 100).
\]
By definition of $\equiv_3$, we thus have $(p_u,s_u) \equiv_3 (p_v,s_v)$.

Similarly, consider now
\begin{eqnarray*}
	u'= 001 01101001 10010110 10010110 0
	  	&=& 001\varphi^3(011) 0,	\\
	v'= 001 0110 10010110 01101001 1001 0
		&=& 001 \varphi^2 (0) \varphi^{3}(10) \varphi^2(1) 0.			
\end{eqnarray*}
They admit a unique factorization of order 3 that are respectively 
\[
	(p_{u'},s_{u'}) = (001,0)
	\quad \text{and} \quad 
	(p_{v'},s_{v'}) = (001\varphi^2(0),\varphi^2(1)0), 
\]
so that we again have $(p_{u'},s_{u'}) \equiv_3 (p_{v'},s_{v'})$
\end{example}

The next result is a direct consequence of Lemma~\ref{lem:relationCuts}.

\begin{corollary}
If a factor of $\bt$ has two distinct factorizations of order~$k$, then these two are equivalent for $\equiv_k$.
\end{corollary}

\begin{definition}[Type of order~$k$]
Given a factor $u$ of $\bt$ of length at least $2^k-1$, the {\em type of order~$k$} of $u$ is the equivalence class of a factorization of order~$k$ of $u$.
We also let $(p_u,s_u)$ denote the factorization of order~$k$ of $u$ for which $|p_u|$ is minimal (we assume that $k$ is understood from the context).
Therefore, two words $u$ and $v$ have the same type of order $k$ if and only if 
\[
	(p_u,s_u) \equiv_k (p_v,s_v).
\]
\end{definition}

\begin{theorem}
\label{thm:typediff}
Let $u,v$ be factors of $\mathbf{t}$ of length $n \geq 2^k-1$. We have
\[
u \sim_k v \Leftrightarrow (p_u,s_u) \equiv_k (p_v,s_v).
\]
\end{theorem}

The condition is trivially sufficient using Lemma~\ref{lem:cancel} and Lemma~\ref{lem:equiv2}. 
For instance, applying several times these two lemmas, we obtain $\varphi^3(01) \sim_3 \varphi^3(11)$,  thus $\varphi^2(0) \varphi^3(01) \sim_3 \varphi^2(0) \varphi^3(11) \sim_3 \varphi^3(11) \varphi^2(0)  $ and finally $u \sim_3 v$ for the words of Example~\ref{ex:type}.

The proof that the condition is necessary is done in Section~\ref{sec:comlexity}.
Preliminary to this, we consider the case of words $u,v$ that do not have any non-empty common prefix of suffix and split the result into two lemmas: either $|p_u| \not\equiv |p_v| \pmod {2^{k-1}}$ (Lemma~\ref{lem:nonlengthmod}) or, $|p_u| \equiv |p_v| \pmod{2^{k-1}}$ (Lemma~\ref{lem:lengthMod}).
We end the section with Lemma~\ref{lem:deletePref} that permits us to deal with factors having some common prefix or suffix.

\begin{lemma}
\label{lem:nonlengthmod}
Let $u,v$ be factors of $\bt$ of length $n \geq 2^k-1$ with no non-empty common prefix or suffix. 
If $(p_u,s_u),(p_v,s_v)$ satisfy $|p_u|+|s_u| < |u|$, $|p_v|+|s_v| < |v|$ and $|p_u| \not\equiv |p_v| \pmod{2^{k-1}}$, then $u \nsim_k v$.
\end{lemma}

\begin{proof}
    The assumptions $|p_u|+|s_u| < |u|$ and $|p_v|+|s_v| < |v|$ imply there exist non-empty words $z,z'$ such that
$$u=p_u\, \varphi^k(z)\, s_u\quad\text{ and }\quad v=p_v\, \varphi^k(z')\, s_v\ .$$
Let $x\in\{u,v\}$. If $p_x=s_x=\varepsilon$, set $j_x:=k$. Otherwise, define $j_x$ as the largest integer such that
$|x|\equiv 0\pmod{2^{j_x-1}}$ and $|p_x|$ or $|s_x|$ is congruent to
$2^{j_x-1}$ modulo $2^{j_x}$. 
In that case, such a $j_x\ge 1$ exists because $p_x$
(resp., $s_x$) is a suffix (resp., prefix) of $\varphi^k(a)$ for some
letter $a$: so it is of the form $p_x=\varphi^{i_r}(a_r)\cdots\varphi^{i_2}(a_2)\varphi^{i_1}(a_1)$ with $i_r<\cdots<i_2<i_1<k$ (resp., $s_x=\varphi^{i_s'}(a_s')\cdots\varphi^{i_2'}(a_2')\varphi^{i_1'}(a_1')$ with $k>i_s'>\cdots >i_2'>i_1'$). More precisely, we have
$$x= \varphi^{i_r}(a_r)\cdots\varphi^{i_2}(a_2)\varphi^{i_1}(a_1) \varphi^k(z) \varphi^{i_s'}(a_s')\cdots\varphi^{i_2'}(a_2')\varphi^{i_1'}(a_1')$$ for some word $z$ 
and $j_x=1+\min\{i_r,i_1'\}$.

Let $j=\min\{j_u,j_v\}$. Observe that $j\le k-1$. 
First, since $|p_u| \not\equiv |p_v| \pmod{2^{k-1}}$, we cannot have $p_u=p_v=s_u=s_v=\varepsilon$.
Moreover, proceed
by contradiction and assume that $j=k$, i.e., $j_u=j_v=k$. In that case, since
$|u|\equiv 0\pmod{2^{k-1}}$, the fact that $|p_u|$ or $|s_u|$ is
congruent to $2^{k-1}$ modulo $2^{k}$ implies that $|u|, |p_u|, |s_u|$
are all congruent to $0$ modulo $2^{k-1}$. The same conclusion holds
for $v$ contradicting the assumption
$|p_u| \not\equiv |p_v| \pmod{2^{k-1}}$.

We will prove that $u\nsim_{j+1} v$. We have two main cases to discuss. Since $u$ and $v$ have the same length and $|u|,|v|\equiv 0\pmod{2^{j-1}}$, we have either $|u|=|v|\equiv 2^{j-1}\pmod{2^j}$ or, $|u|=|v|\equiv 0\pmod{2^j}$.

The first case is split into three sub-cases.
\begin{itemize}
  \item[1.1)] Since $u$ and $v$ have no common prefix, we can first assume that $u=\varphi^{j-1}(0) \varphi^j(u')$ and $v=\varphi^{j-1}(1)\varphi^j(v')$ for some words $u',v'$ (we can exchange the roles of $0$ and $1$). The conclusion $u\nsim_j v$ follows directly from Lemma~\ref{lem:noneq} because $\varphi^{j-1}(0)\nsim_j\varphi^{j-1}(1)$ by Lemma~\ref{lem:equiv-not}.
  \item[1.2)] Consider the case where $u=\varphi^{j-1}(0) \varphi^j(u')$ and $v=\varphi^j(v')\varphi^{j-1}(1)$ for some words $u',v'$. We can make use of Lemma~\ref{lem:equiv2}, $v\sim_j \varphi^{j-1}(1)\varphi^j(v')$ and conclude as in the previous case. 
  \item[1.3)] The last sub-case is when
    $u=\varphi^{j-1}(0) \varphi^j(u')=\varphi^{j-1}(0\varphi(u'))$ and
    $v=\varphi^j(v')\varphi^{j-1}(0)=\varphi^{j-1}(\varphi(v')0)$
    (the situation with $1$ instead of $0$ can be treated similarly). 
If $j=1$, we have $\binom{u}{01} - \binom{v}{01} = |u'|>0$.    
    We will assume $j>1$.    
    Consequently,  $|u|=2^{j-1}+2^j |u'|$ is even, thus $|u|\ge 2^k$ and $|u'|\ge 2^{k-j}\ge 2$.  From
    Remark~\ref{rem:difference} where multiplicities $m(x)$ are here related to $f^{j-1}(01^j)$, we get
    \begin{eqnarray*}
\binom{u}{01^j}-\binom{v}{01^j}&=&
\sum_{x \in f^{j-1}(01^j)} m(x) \left[ \binom{0\varphi(u')}{x} - \binom{\varphi(v')0}{x} \right].
\end{eqnarray*}
Recall that $f^{j-1}(01^j) $ only contains elements in $A^{\leq 2}$. In the above formula, only $x=01$ and $x=10$ will give non-zero terms. 
Compute
\begin{eqnarray*}
	\binom{0\varphi(u')}{01} - \binom{\varphi(v')0}{01}
	&=& |u'|+\binom{\varphi(u')}{01}-\binom{\varphi(v')}{01}	\\
	&=& |u'|+\binom{|u'|}{2}+\binom{u'}{0}-\binom{|v'|}{2}-\binom{v'}{0}.
\end{eqnarray*}
Hence,
\begin{eqnarray*}
\binom{u}{01^j}-\binom{v}{01^j}
&=&
\left(m(01)-m(10)\right)|u'|\\
&&+
m(01)\left( \binom{u'}{0}-\binom{v'}{0}\right) + m(10) \left(
    \binom{u'}{1}-\binom{v'}{1}\right) \\
 &=& 
(m(01)-m(10))\left(|u'|+ \binom{u'}{0}-\binom{v'}{0}\right).
\end{eqnarray*}
The last equality comes from the fact that $\binom{u'}{0} - \binom{v'}{0} = \binom{v'}{1} - \binom{u'}{1}$ because $|u'|=|v'|$.

Since $u'$ and $v'$ are factors of $\bt$ of the same length, it is clear that $\binom{u'}{0} - \binom{v'}{0} \in \{-2,-1,0,1,2\}$. However, in this sub-case the value $-2$ is not realized, since $v'$ starts with $1$ (because $u$ and $v$ have no common prefix). Thus, by Proposition~\ref{pro:diffpower},
\[
\binom{u}{01^j} - \binom{v}{01^j} \geq (m(01)-m(10)) (|u'|-1) > 0
\]
and $u \nsim_{j+1} v$.

\end{itemize}

For the second case, we assume that $|u|=|v|\equiv 0\pmod{2^j}$. We have four sub-cases for which we know that $|u'|\geq 1$.
\begin{itemize}
  \item[2.1)] If 
$u=\varphi^{j-1}(0)\varphi^j(u')\varphi^{j-1}(0)$ and 
$v=\varphi^{j}(v')$, then we know that $v'$ is of the form $1v''$ because $u$ and $v$ have no common prefix.
We have $u\sim_j \varphi^{j-1}(0)\varphi^{j-1}(0)\varphi^j(u')$ and $v=\varphi^{j-1}(1)\varphi^{j-1}(0)\varphi^j(v'')$ so we can directly conclude that $u\nsim_j v$ applying Lemma~\ref{lem:noneq} and Lemma~\ref{lem:equiv-not}.

  \item[2.2)]  
If $u=\varphi^{j-1}(0)\varphi^j(u')\varphi^{j-1}(1)=\varphi^{j-1}(0\varphi(u')1)$ and $v=\varphi^{j}(v')$, we know that $v'$ starts with $1$ and ends with $1$ because $u$ and $v$ have no common prefix or suffix. We have
 \begin{eqnarray*}
\binom{u}{01^j}-\binom{v}{01^j}&=&
m(01)\left( \binom{0\varphi(u')1}{01}-\binom{\varphi(v')}{01}\right) + m(10) \left(
    \binom{0\varphi(u')1}{10}-\binom{\varphi(v')}{10}\right)\\
&=&m(01) \left[ 1+2|u'| + \binom{u'}{0} - \binom{v'}{0} + \binom{|u'|}{2} - \binom{|v'|}{2} \right] 
\\&&+ m(10) \left[ \binom{|u'|}{2} - \binom{|v'|}{2} + \binom{u'}{1} - \binom{v'}{1} \right]
    \end{eqnarray*}
    
   Here, $|u'|=|v'|-1$, so $\binom{u'}{1} - \binom{v'}{1} = \binom{v'}{0} - \binom{u'}{0} - 1$, $\binom{|u'|}{2}-\binom{|v'|}{2} = -|u'|$ and we obtain
   \[
   \binom{u}{01^j} - \binom{v}{01^j} = (m(01)-m(10)) \left(1+ |u'| + \binom{u'}{0} - \binom{v'}{0} \right).
   \]    
    
We need to characterize the values that can be taken by $\binom{u'}{0}-\binom{v'}{0}$. Two cases may happen: if $|u'|$ is even, there exists $\ell > 0$ such that $|u'|=2 \ell$. In this case, $|v'|=2 \ell +1$. Since $v'$ begins and ends with a $1$, $\binom{v'}{0} = \ell$. Therefore,
\[
\binom{u'}{0} - \binom{v'}{0}  \in \{-1,0,1 \}.
\] 
If $|u'|$ is odd, there exists $\ell$ such that $|u'|=2 \ell +1$ and $|v'|=2\ell+2$. For the same reason as above, we cannot have $\binom{v'}{0} = \ell+2$ and $\binom{u'}{0} - \binom{v'}{0}$ takes the same values.
We thus have, in both cases, $\binom{u}{01^j} - \binom{v}{01^j} > 0$.

  \item[2.3)] Now assume that $u=\varphi^{j-1}(0)\varphi^j(u')\varphi^{j-1}(0)$ and 
$v=\varphi^{j-1}(1)\varphi^j(v')\varphi^{j-1}(1)$. Since $|0u'0|_0\neq |1v'1|_0$, when applying Remark~\ref{rem:difference} all words in $A^{\le 2}$ are contributing and we obtain
 \begin{eqnarray*}
\binom{u}{01^j}-\binom{v}{01^j}&=& 
2(m(0)-m(1))+ (m(00)-m(11)) (2|u'|+1) \\
&&+ m(01)\left( \binom{u'}{0}-\binom{v'}{0}\right) + m(10) \left(
    \binom{u'}{1}-\binom{v'}{1}\right) \\
    &=& 2(m(0)-m(1))+ (m(00)-m(11)) \left(2|u'|+1+\binom{u'}{0}-\binom{v'}{0}\right)
\end{eqnarray*}
where the last equality comes from Proposition \ref{pro:diffpower}. One can again conclude in the same way, making use of Proposition~\ref{pro:diffpower2}.

  \item[2.4)] The last case is when
$u=\varphi^{j-1}(0)\varphi^j(u')\varphi^{j-1}(1)$ and 
$v=\varphi^{j-1}(1)\varphi^j(v')\varphi^{j-1}(0)$. We have 
 \begin{eqnarray*}
\binom{u}{01^j}-\binom{v}{01^j}&=&
m(01)\left( \binom{0\varphi(u')1}{01}-\binom{1\varphi(v')0}{01}\right) + m(10) \left(
    \binom{0\varphi(u')1}{10}-\binom{1\varphi(v')0}{10}\right)\\
&=& (m(01)-m(10)) (2|u'|+1) \\
&& +m(01)\left( \binom{u'}{0}-\binom{v'}{0}\right) + m(10) \left(
    \binom{u'}{1}-\binom{v'}{1}\right) \\
&=& (m(01)-m(10)) \left(2|u'|+1+\binom{u'}{0}-\binom{v'}{0}\right) \\ &\geq& (m(01)-m(10)) (2|u'|-1) >0
\end{eqnarray*}
with the same reasoning as above.
\end{itemize}
\end{proof}

\begin{lemma}
\label{lem:lengthMod}
Let $u,v$ be factors of $\bt$ of length $n \geq 2^k-1$ with no non-empty common prefix or suffix. If $(p_u,s_u) \not\equiv_k (p_v,s_v)$ with $|p_u|\equiv|p_v| \pmod{2^{k-1}}$, then $u \nsim_k v$.
\end{lemma}

\begin{proof}
     Let $\ell$ (resp., $\ell'$) be the greatest integer less than $k$ such that $|p_u|\equiv 0\pmod{2^\ell}$ (resp., $|s_u|\equiv 0\pmod{2^{\ell'}}$). The assumption $|p_u|\equiv|p_v| \pmod{2^{k-1}}$ implies that $|s_u|\equiv|s_v| \pmod{2^{k-1}}$ and thus, $|p_u|\equiv|p_v| \pmod{2^{\ell}}$ and $|s_u|\equiv|s_v| \pmod{2^{\ell'}}$. We have three cases to take into account.
    \begin{itemize}
      \item[1)] If $\ell<\ell'$ (the case $\ell'<\ell$ is symmetric taking the reversal of the words), then $|s_u|$ and $|s_v|$ are even multiples of $2^\ell$, i.e., there exist $x, x'\in A^*$ such that $s_u=\varphi^{\ell+1}(x)$ and $s_v=\varphi^{\ell+1}(x')$. Moreover, by maximality of $\ell$, $|p_u|$ and $|p_v|$ are odd multiples of $2^\ell$, i.e., there exist $a\in A$ and $y,y'\in A^*$ such that 
$$p_u=\varphi^\ell (a)\varphi^{\ell+1}(y), \;\;\; ~p_v=\varphi^\ell (\overline{a})\varphi^{\ell+1}(y')$$ 
hence
$$u= p_u \varphi^{\ell+1}(z) s_u=\varphi^\ell (a)\varphi^{\ell+1}(y) \varphi^{\ell+1}(z) \varphi^{\ell+1}(x)$$
and 
$$v= p_v \varphi^{\ell+1}(z') s_v = \varphi^\ell (\overline{a})\varphi^{\ell+1}(y') \varphi^{\ell+1}(z') \varphi^{\ell+1}(x')$$
for some $z,z'$. As usual, by Lemma~\ref{lem:equiv2}, we can conclude because $\varphi^\ell(a)\nsim_{\ell+1}  \varphi^\ell (\overline{a})$, $|yzx|=|y'z'x'|$ and $\ell+1\leq\ell'\leq k-1$.
\item[2)] If $\ell=\ell'=k-1$, we have to distinguish the cases where $p_u$ or $s_u$ are empty. 

\begin{itemize}
\item
If $p_u = \varepsilon = s_u$, we have neither $p_v=\varepsilon=s_v$ nor, $p_v = \varphi^{k-1}(a), s_v = \varphi^{k-1}(\overline{a})$ because $u$ and $v$ do not have the same type of order~$k$. This implies that $v$ is of the form $\varphi^{k-1}(a)\varphi^{k}(z)\varphi^{k-1}(a)$ and we can conclude that $u\nsim_k v$. Indeed, since $z=\overline{a}z''$ (recall that $u$ and $v$ have no common prefix), then  
\begin{align*}
u = \varphi^{k}(z) \sim_k \varphi^{k-1}(\overline{a}) \varphi^{k-1}(a)\varphi^{k}(z'') &\sim_k  \varphi^{k-1}(\overline{a}) \varphi^{k}(z'') \varphi^{k-1}(a) \\& \nsim_k  \varphi^{k-1}(a) \varphi^{k}(z')  \varphi^{k-1}(a) = v.
\end{align*}

\item
If $p_u=\varepsilon$ and $s_u = \varphi^{k-1}(a)$ (or the opposite), the fact that $(p_u,s_u) \not\equiv_k (p_v,s_v)$ gives us the possibilities $v=\varphi^{k-1}(\overline{a})\varphi^{k}(z')$ or, $v=\varphi^{k}(z')\varphi^{k-1}(\overline{a})$ for some $z'$. But obviously $u \nsim_k v$.

\item 
If $p_u=\varphi^{k-1}(a)$ and $s_u=\varphi^{k-1}(b)$, then $u=\varphi^{k-1}(a)\varphi^{k}(z') \varphi^{k-1}(b)$ and 
$v=\varphi^{k-1}(\overline{a})\varphi^{k}(z') \varphi^{k-1}(\overline{b})$ for some $z,z'$. Moreover $a=b$ because $u$ and $v$ do not have the same type of order~$k$.
\\ Let us assume that $a=b=0$.
Since $\varphi^{k-1}(0)\nsim_k\varphi^{k-1}(1)$, there exists a word $w$ of length $k$ such that $$\binom{\varphi^{k-1}(0)}{w}\neq \binom{\varphi^{k-1}(1)}{w}.$$
Therefore, we get
\begin{align*}
\binom{u}{w}-\binom{v}{w} 
&= \sum_{\substack{r,s,t \in A^* \\ rst=w}} \left[ \binom{\varphi^{k-1}(0)}{r}\binom{\varphi^{k}(z)}{s} \binom{\varphi^{k-1}(0)}{t} - \binom{\varphi^{k-1}(1)}{r}\binom{\varphi^{k}(z')}{s}\binom{\varphi^{k-1}(1)}{t} \right].
\end{align*}
In the above sum, every term such that $|r|<k$ and $|t|<k$ vanishes because $\varphi^{k-1}(0)\sim_{k-1}\varphi^{k-1}(1)$. Hence, we get
\begin{align*}
\binom{u}{w}-\binom{v}{w} 
&= 2\left[ \binom{\varphi^{k-1}(0)}{w}-\binom{\varphi^{k-1}(1)}{w}\right]\neq 0.
\end{align*}

\end{itemize}

\item[3)] Now assume $\ell=\ell'<k-1$. Hence, there exist letters $a,b$ and words $z,z'$ such that
$$u=\varphi^\ell(a)\varphi^{\ell+1}(z)\varphi^\ell(b) \quad\text{ and }\quad
v=\varphi^\ell(\overline{a})\varphi^{\ell+1}(z')\varphi^\ell(\overline{b}).$$
If $a=b$, then we can conclude that $u \nsim_k v$ as in the last part of case 2). Assume that $a\neq b$ (and $a=0$, $b=1$). Then compute (the reader should be used to this kind of computations)
\begin{eqnarray*}
&&\binom{u}{01^\ell}-\binom{v}{01^\ell}\\
&=&\left(m_{f^{\ell-1}(01^\ell)}(01)-m_{f^{\ell-1}(01^\ell)}(10)\right) \left[ 
1+2|z|+\binom{\varphi(z)}{01}-\binom{\varphi(z')}{01}
\right] \\
&=& \left(m_{f^{\ell-1}(01^\ell)}(01)-m_{f^{\ell-1}(01^\ell)}(10)\right)  \left[ 1+2|z|+ \binom{|z|}{2} + |z|_0  - \binom{|z'|}{2} - |z'|_0 \right]
\end{eqnarray*}
which is positive since $|z| = |z'|$.
    \end{itemize}
\end{proof}

When deleting common prefixes and suffixes of two factors with different types of order~$k$, if the resulting factors are long enough, their types of order~$k$ are different.

\begin{lemma}
\label{lem:deletePref}
Let $u$ and $v$ be factors of $\textbf{t}$ of the same length which do not have the same type of order~$k$. Let $x$ (resp., $y$) be the longest common prefix (resp., suffix) of $u$ and $v$, i.e., $u=xu'y$ and $v=xv'y$. If $|u'y| \geq 2^k-1$ then $u'y$ and $v'y$ do not have the same type of order~$k$. Similarly, if $|xu'| \geq 2^k-1$ then $xu'$ and $xv'$ do not have the same type of order~$k$.
 \end{lemma}
\begin{proof}
We only show the result for $u'y$ and $v'y$. Let us assume that $x \neq \varepsilon$.

Let $D\in \Cut_k(u)$ such that $|p_u|=\min D$. There exists $C\in \Cut_k(u'y)$ such that $C+|x|\subset D$. In particular, $|x|+\min C\equiv \min D\pmod{2^k}$. 
There exists $C'\in \Cut_k(u'y)$ such that $|p_{u'y}|=\min C'$. From Lemma~\ref{lem:relationCuts}, we know that $\min C\equiv \min C'\pmod{2^{k-1}}$. Hence 
$$|xp_{u'y}|=|x|+\min C' \equiv |x|+\min C \pmod{2^{k-1}}$$
and we conclude that $|xp_{u'y}|\equiv |p_u|\pmod{2^{k-1}}$. Otherwise stated, $$|p_u|\equiv |p_v| \pmod{2^{k-1}}\quad\text{ if and only if }|p_{u'y}|\equiv |p_{v'y}| \pmod{2^{k-1}}\, .$$
Using that fact, if $|p_u| \not\equiv |p_v| \pmod{2^{k-1}}$ then $u'y$ and $v'y$ do not have the same type of order~$k$ (see Remark~\ref{rem:lgequivk}). In what follows, we may assume that $|p_u| \equiv |p_v| \pmod{2^{k-1}}$ and thus, $|s_u| \equiv |s_v| \pmod{2^{k-1}}$.

We have two main cases to discuss: either $|p_u | = |p_v|$ or, $| |p_u| - |p_v| | = 2^{k-1}$.

\begin{itemize}
  \item[1)] Assume $|p_u | = |p_v|$. This implies that $p_u = p_v$. Indeed these two words are suffixes of the same length of a word of the form $\varphi^{k}(a)$ (where $a$ is a letter). Since they share a common prefix ($x\neq\varepsilon$), they must be equal. 
Consequently, we have $s_u \neq s_v$, otherwise $u$ and $v$ would have the same type. 
Therefore, $y = \varepsilon$ and we will write $u'$ instead of $u'y$. Let us show that $(p_{u'},s_{u'}) = (\varepsilon, s_u)$ and $(p_{v'},s_{v'})=(\varepsilon,s_v)$ meaning that $u'y$ and $v'y$ do not have the same type. The words $u$ and $v$ are respectively of the form $p_u \varphi^k(u'') s_u$ and $p_v \varphi^k(v'') s_v$ for some words $u'', v''$.   
 
Since $p_u = p_v$, we have $|x| \geq |p_u|$. 
If $|x|=|p_u|$, $u'= \varphi^k(u'') s_u$ and $v'= \varphi^k(v'') s_v$ so $(p_{u'},s_{u'}) = (\varepsilon, s_u)$ and $(p_{v'},s_{v'})=(\varepsilon,s_v)$. Otherwise, $|x|>|p_u|$ and there exists $\ell > 0$ such that\footnote{We assume that a finite word $u''$ has its first symbol indexed by $1$, so $u''_{[1,j]}$ denotes the prefix of $u''$ of length $j$.} $p_u \varphi^k (u''_{[1, \ell-1]})$ is a proper prefix of $x$ and such that $x$ is a prefix of $p_u \varphi^k (u''_{[1, \ell]})$. Then $x = p_u \varphi^k (u''_{[1, \ell]})$. This is due to the fact that if $\varphi^k(a)$ and $\varphi^k(b)$ share a non-empty common prefix, then $a=b$.
Thus, $u'=\varphi^k(u''_{[\ell+1, |u''|]}) s_u$, $v'=\varphi^k(v''_{[\ell+1, |v''|]}) s_v$ and we are done.

\item[2)] Let us consider the second case and assume that $|p_v| = |p_u| + 2^{k-1}$. As usual, $u$ and $v$ are of the form $p_u \varphi^k(u'') s_u$, $p_v \varphi^k(v'') s_v$ and set $u'''=\varphi(u'')$ and $v'''=\varphi(v'')$. Let $a$ be the letter such that $p_v$ is a suffix of $\varphi^k (a)$. 
Two sub-cases have to be considered: either $|s_v|=|s_u|+2^{k-1}$ or, $|s_u|=|s_v|+2^{k-1}$. In each one of them, we let the letter $b$ be such that the longest word in $\{s_u,s_v\}$ is a prefix of $\varphi^k (b)$.

\item[2.a)] Consider the first sub-case, $|s_v|=|s_u|+2^{k-1}$. By definition of $b$, $s_v$ is a prefix of $\varphi^{k}(b)$. We have $p_v= w_1 \varphi^{k-1}(\overline{a})$ and $s_v=\varphi^{k-1}(b) w_2$ where $w_1$ (resp., $w_2$) is a suffix (resp., prefix) of $\varphi^{k-1}(a)$ (resp., $\varphi^{k-1}(\overline{b})$).
Recall that $u$ and $v$ have a non-empty common prefix $x$. Since  $|s_v|=|s_u|+2^{k-1}<2^k$, we get $|p_u| < 2^{k-1}$ and $p_u$ is a suffix of $\varphi^{k-1}(a)$. Hence, $p_u=w_1$. Figure~\ref{fig:cas1} illustrates the situation.

\begin{figure}[h!]
\begin{center}
\begin{tikzpicture}
\draw (0,0) rectangle (8,0.5);
\draw (1,0) -- (1,0.5);
\draw (7.2,0) -- (7.2,0.5);
\draw [color=gray] (1.05,0.05) rectangle (3.95,0.45);
\draw [dashed,color=gray] (2.5,0.05) -- (2.5,0.45);

\node (0) at (0.5,0.25) {\footnotesize{$w_1$}};
\node (1) at (-1,0.25) {$u$ :};
\node (2) at (7.6,0.25) {\footnotesize{$s_u$}};

\draw (0,-1.5) rectangle (8,-1);
\draw (2.5,-1.5) -- (2.5,-1);
\draw (5.7,-1.5) -- (5.7,-1);
\draw [dashed,color=gray] (7.2,-1.5) -- (7.2,-1);
\draw [dashed,color=gray] (1,-1.5) -- (1,-1);

\node (0) at (0.5,-1.25) {\footnotesize{$w_1$}};
\node (0) at (1.8,-1.25) {\footnotesize{$\varphi^{k-1}(\overline{a})$}};
\node (1) at (-1,-1.25) {$v$ :};
\node (0) at (6.5,-1.25) {\footnotesize{$\varphi^{k-1}(b)$}};
\node (0) at (7.6,-1.25) {\footnotesize{$w_2$}};

\draw [decorate,decoration={brace,amplitude=10pt,mirror}] (5.7,-1.6) -- (8.7,-1.6);
\node (3) at (7.2,-2.3) {$\varphi^k (b)$};
\draw [decorate,decoration={brace,amplitude=10pt,mirror}]
(-0.5,-1.6) -- (2.5,-1.6);
\node (4) at (1,-2.3) {$\varphi^k (a)$};

\draw [color=red] (0,-0.1) -- (0,-0.3);
\draw [color=red] (0,-0.2) -- (1,-0.2);
\draw [dashed,color=red] (1,-0.2) -- (1.5,-0.2);

\draw [color=red] (0,-0.7) -- (0,-0.9);
\draw [color=red] (0,-0.8) -- (1,-0.8);
\draw [dashed,color=red] (1,-0.8) -- (1.5,-0.8);
\node [color=red] (5) at (0.7,-0.5) {$x$};

\end{tikzpicture}
\end{center}
\caption{Decomposition of $u$ and $v$ in the first sub-case.}
\label{fig:cas1}
\end{figure}
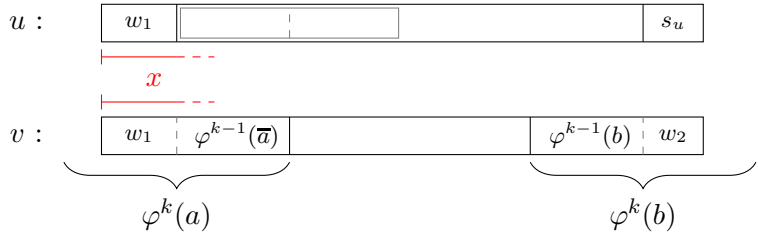

The word $s_u$ is a prefix of some $\varphi^k(c)$. Hence $w_2 = \overline{s_u}$ or $w_2 = s_u$ depending on whether $y$ is empty or not. Since $u$ and $v$ do not have the same type, $w_2=\overline{s_u}$ and these words are non-empty or, $\overline{a} = b$.

Using the same argument as before, $|x| \geq |p_u|$. 
If $|x| = |p_u|$, we have $u'y = \varphi^{k} (u'') s_u$ and $v'y= \varphi^{k-1}(\overline{a}) \varphi^k (v'') \varphi
^{k-1}(b) w_2$ and comparing the pairs $(\varepsilon,s_u)$ and $(\varphi^{k-1}(\overline{a}),\varphi^{k-1}(b) w_2)$, we conclude that $u'y$ and $v'y$ do not have the same type (whenever $w_2=\overline{s_u}\neq\varepsilon$ or, $\overline{a} = b$). 

Otherwise, $|x|>|p_u|$ and there exists some $\ell > 0$ such that 
\begin{align}
\label{equ:x}
 x = p_u \varphi^{k-1}(u'''_{[1, \ell]}) = p_u \varphi^{k-1}(\overline{a}) \varphi^{k-1}(v'''_{[1, \ell-1]}). 
 \end{align}

We thus have
\[ u'y = \varphi^{k-1} (u'''_{[\ell+1, |u'''|]}) s_u \quad\text{ and }\quad  v'y=\varphi^{k-1} (v'''_{[\ell, |v'''|]}) s_v. \]
 
From equalities (\ref{equ:x}), we observe that
\begin{align}
\label{equ:utierce1}
u'''_i=\left\{\begin{array}{ll}
\overline{a}, & \text{if } i=1;\\
v'''_{i-1}, & \text{if } 1<i\le\ell.
\end{array}\right.
\end{align}
Moreover, since $u'''_{2i+1} u'''_{2i+2} = \varphi (u''_{i+1})$, for all $i\in\{1,\ldots,|u''|\}$, we have $u'''_{2i+1} = \overline{u'''_{2i+2}}$. Similarly, we have $v'''_{2i+1} = \overline{v'''_{2i+2}}$. We may thus conclude that
\begin{align}
\label{equ:utierce2}
u'''_i=\left\{\begin{array}{lc}
\overline{a}, & \text{if $i$ is odd};\\
a, & \text{if $i$ is even}.
\end{array}\right.
\end{align}
If $\ell$ is even, we have 
\[u'y =  \varphi^k (u''_{[\frac{\ell+2}{2},|u''|]}) s_u \quad
\text{ and } 
\quad v'y = \varphi^{k-1}(v'''_{\ell}) \varphi^k (v''_{[\frac{\ell+2}{2},|v''|]}) \varphi^{k-1}(b) w_2\, . \]
Thus, $(p_{u'y}, s_{u'y})= (\varepsilon, s_u)$ and $(p_{v'y}, s_{v'y})= (\varphi^{k-1}(\overline{a}), \varphi^{k-1}(b) w_2)$ and $u'y$, $v'y$ do not have the same type of order~$k$.

If $\ell$ is an odd number,
\[
u'y =  \varphi^{k-1}(u'''_{\ell+1}) \varphi^k (u''_{[\frac{\ell+3}{2},|u''|]}) s_u
\quad \text{ and }\quad
v'y =  \varphi^k (v''_{[\frac{\ell+1}{2},|v''|]}) \varphi^{k-1}(b) w_2\, .
\]
In that case, $(p_{u'y}, s_{u'y})= (\varphi^{k-1} (a), s_u)$ and $(p_{v'y}, s_{v'y})= (\varepsilon, \varphi^{k-1}(b) w_2)$ and again, $u'y$, $v'y$ do not have the same type of order~$k$.

\item[2.b)] Let us care about the second sub-case: $|s_u|  = |s_v| + 2^{k-1}$. By definition of $b$, $s_u$ is a prefix of $\varphi^k(b)$. Thus $p_v= w_1 \varphi^{k-1}(\overline{a})$ and $s_u=\varphi^{k-1}(b) w_2$ where $w_1$ (resp., $w_2$) is a suffix (resp., prefix) of $\varphi^{k-1}(a)$ (resp., $\varphi^{k-1}(\overline{b})$).
Otherwise stated, as illustrated in Figure~\ref{fig:cas2}, we have 
\[
u = p_u \varphi^{k} (u'') \varphi^{k-1}(b) w_2
\quad\text{ and }\quad
v= p_u \varphi^{k-1}(\overline{a}) \varphi^{k}(v'') s_v\, .
\]
\begin{figure}[h!]
\begin{center}
\begin{tikzpicture}
\draw (0,0) rectangle (8,0.5);
\draw (1,0) -- (1,0.5);
\draw (7.2,-1) -- (7.2,-1.5);
\draw [color=gray] (1.05,0.05) rectangle (3.95,0.45);
\draw [dashed,color=gray] (2.5,0.05) -- (2.5,0.45);

\node (0) at (0.5,0.25) {\footnotesize{$w_1$}};
\node (1) at (-1,0.25) {$u$ :};
\node (2) at (7.6,0.25) {\footnotesize{$w_2$}};

\draw (0,-1.5) rectangle (8,-1);
\draw (2.5,-1.5) -- (2.5,-1);
\draw (5.7,0) -- (5.7,0.5);
\draw [dashed,color=gray] (7.2,0) -- (7.2,0.5);
\draw [dashed,color=gray] (1,-1.5) -- (1,-1);

\node (0) at (0.5,-1.25) {\footnotesize{$w_1$}};
\node (0) at (1.8,-1.25) {\footnotesize{$\varphi^{k-1}(\overline{a})$}};
\node (1) at (-1,-1.25) {$v$ :};
\node (0) at (6.5,0.25) {\footnotesize{$\varphi^{k-1}(b)$}};
\node (0) at (7.6,-1.25) {\footnotesize{$s_v$}};

\draw [decorate,decoration={brace,amplitude=10pt}] (5.7,0.6) -- (8.7,0.6);
\node (3) at (7.2,1.3) {$\varphi^k (b)$};
\draw [decorate,decoration={brace,amplitude=10pt,mirror}]
(-0.5,-1.6) -- (2.5,-1.6);
\node (4) at (1,-2.3) {$\varphi^k (a)$};

\draw [color=red] (0,-0.1) -- (0,-0.3);
\draw [color=red] (0,-0.2) -- (1,-0.2);
\draw [dashed,color=red] (1,-0.2) -- (1.5,-0.2);

\draw [color=red] (0,-0.7) -- (0,-0.9);
\draw [color=red] (0,-0.8) -- (1,-0.8);
\draw [dashed,color=red] (1,-0.8) -- (1.5,-0.8);
\node [color=red] (5) at (0.7,-0.5) {$x$};

\end{tikzpicture}
\end{center}
\caption{Decomposition of $u$ and $v$ in the second sub-case.}
\label{fig:cas2}
\end{figure}
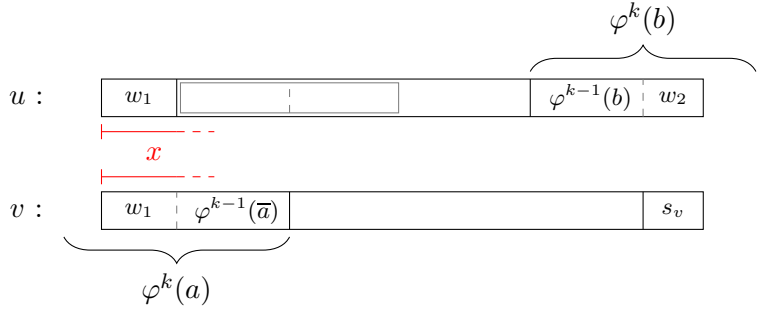

Observe again that $w_2 = s_v$ or $w_2=\overline{s_v}$.

Because $u$ and $v$ do not have the same type of order~$k$, $w_2 = \overline{s_v}$ and these words are non-empty or, $a = b$. 

If $x=p_u$, $(p_{u'y},s_{u'y}) = (\varepsilon, \varphi^{k-1}(b) w_2)$ and $(p_{v'y},s_{v'y}) = (\varphi^{k-1} (\overline{a}), s_v)$. Comparing these two pairs, we get $(p_{u'y}, s_{u'y}) \not\equiv_k (p_{v'y},s_{v'y})$, i.e., $u'y$ and $v'y$ do not have the same type.

Otherwise, as in the previous case, there exists $\ell > 0$ such that $x=p_u \varphi^{k-1}(u'''_{[1,\ell]})$. Observe that equalities (\ref{equ:utierce1}) and (\ref{equ:utierce2}) are still valid. It remains to discuss about the parity of $\ell$ to conclude.
\end{itemize}
\end{proof}

\section{$k$-binomial complexity of the Thue--Morse word}\label{sec:comlexity}

Using the lemmas from the previous section, we first show that two different factors of $\bt$ of length at most $2^k-1$ are never $k$-binomially equivalent.
Then, we take into account factors of length at least $2^k$. On the one hand, we prove that $(p_u,s_u) \not \equiv_k (p_v,s_v)$ implies $u \nsim_k v$. On the other hand, we compute the number of equivalence classes of $\# \left( \{(p_u,s_u) : u \in \Fac_n(\mathbf{t}) \}/\!\equiv_k\right)$.

\begin{proposition}
\label{prop:smallWords}
Let $u,v$ be two different factors of $\bt$ of length $n \leq 2^k-1$, which do not have any common prefix or suffix.  
We have $u \nsim_k v$.
\end{proposition}
\begin{proof}
If $n \leq 3$, this is trivial: take $u$, $v$ two different factors of length $n$, $\binom{u}{u}=1$ and $\binom{v}{u}=0$. Since $k\ge 3$, $u\nsim_k v$. 
Let us assume $n\ge 4$ and set $j=\max \{i \leq k : |u| \geq 2^{i+1} \}$. 
Notice that $1 \leq j \leq k-2$. The type of order~$j$ of $u$ and $v$  is well-defined. Either they have the same type of order~$j$ or, they don't.
If they do not have the same type, since we always have $|p_u|+|s_u| \leq 2^{j+1} -2$, we have $|p_u|+|s_u| < |u|$. The same holds for $v$. 
By applying Lemmas~\ref{lem:nonlengthmod} or~\ref{lem:lengthMod}, we obtain that $u \not \sim_j v$, thus $u \nsim_k v$.

We can thus assume that $(p_u,s_u) \equiv_j (p_v,s_v)$.
Let us consider the different cases of Definition \ref{def:type}.
Since $u$ and $v$ do not share any common prefix or suffix, this gives restrictions to the different possibilities. For instance, in the situation (1.a) of the definition, we get $(p_u,s_u) = (p_v,s_v) = (\varepsilon,\varepsilon)$.
For all what remains, let $a,b$ be two different letters.

\begin{itemize}
\item[(1.a)]
If $(p_u,s_u)=(p_v,s_v)$, then $p_u = p_v = s_u = s_v = \varepsilon$, so $u = \varphi^j(u')$ and $v = \varphi^j(v')$ for some words $u'$ and $v'$ of length 2 or 3 (by definition of $j$).
Since $u$ and $v$ do not have any non-empty common prefix and suffix, $u'$ and $v'$ have distinct first and last letter.
Recalling that $\bt$ is cube-free, we thus have to consider the cases
\[
(u',v') \in \{(aa,bb),(aab,bba),(aba,bab),(ab,ba),(aab,baa) \mid a,b \in A, a \neq b\}.
\]
Using the same kind of computations as before, e.g., in the proof of Lemma~\ref{lem:nonlengthmod} making use of Remark~\ref{rem:difference}, we get 
\[
\binom{u}{ab^j} - \binom{v}{ab^j} = \binom{\varphi^j(u')}{ab^j} - \binom{\varphi^j(v')}{ab^j} =  \left(m_{f^{j-1}(ab^j)}(ab)-m_{f^{j-1}(ab^j)}(ba)\right)(|u'|_a - |v'|_a)
\]
which is positive for the first three pairs $(u',v')$. For the last two ones, simply compute
\begin{align*}
\binom{u}{ab^{j+1}} - \binom{v}{ab^{j+1}} =& \, m_{f^{j}(ab^{j+1})}(ab)\left( \binom{u'}{ab} - \binom{v'}{ab} \right) \\ &-m_{f^{j}(ab^{j+1})}(ba) \left( \binom{u'}{ba} - \binom{v'}{ba} \right) >0.
\end{align*}

Since $j \leq k-2$, this implies in all cases that $u \nsim_k v$.

\item[(1.b) and (1.c)]
Let us consider the case where there exists some letter $a$ such that $u = \varphi^{j-1}(a)\varphi^j(u')$ and $v = \varphi^j(v') \varphi^{j-1}(a)$ (or the converse). We know that the first letter of $v'$ is $b$ and the last letter of $u'$ is $a$ (because $u$ and $v$ have distinct first and last letter).
Since $\bt$ is overlap-free, we have to consider the cases
\begin{align*}
	(u',v') \in 
	&\{ (aa,ba),(aa,bb),(ba,ba),(ba,bb),\\
	&(aba,bab),(aba,bba),
	(baa,bab),(baa,bba)  \mid a,b \in A, a \neq b \}.
\end{align*}
Indeed, $u'$ cannot be equal to $bba$ because otherwise $u$ would be equal to $\varphi^{j-1}(ababaab)$. Similarly, $v'$ cannot be equal to $baa$. 
For every pairs in the above set, compute 
\[
\binom{\varphi
^{j-1}(a\varphi(u'))}{ab^j} - \binom{\varphi^{j-1}(\varphi(v')a)}{ab^j} = (m_{f^{j-1}(ab^j)}(ab)-m_{f^{j-1}(ab^j)}(ba)) (|u'| + |u'|_a- |v'|_a) > 0.
\]

\item[(1.d)]
Assume now that there exists a letter $a$ such that $u=\varphi^{j-1}(a) \varphi^j(u') \varphi^{j-1}(b)$ and $v=\varphi^{j-1}(b) \varphi^j(v') \varphi^{j-1}(a)$. We have $|u|=2^j+2^j|u'|$ and from the definition of $j$, we get $|u'|\le 2$. Then compute
\[
\binom{\varphi
^{j-1}(a\varphi(u')b)}{ab^j} - \binom{\varphi^{j-1}(b\varphi(v')a)}{ab^j} = (m_{f^{j-1}(ab^j)}(ab)-m_{f^{j-1}(ab^j)}(ba)) (1+2|u'| + |u'|_a- |v'|_a).
\]
This quantity is positive for every $u',v' \in A \cup A^2$.

\item[(2.a) and (2.b)]
Otherwise, there exists some letter $a$ such that $u=\varphi^{j-1}(a) \varphi^{j}(u') \varphi^{j-1}(b)$ and $v=\varphi^j(v')$. Again $|u'|\le 2$ and $|v'|=|u'|+1$. Then, $v'$ has to begin with $a$ and end with $b$. The cases to consider are the following ones: 
\begin{align*}
(u',v') \in &\{(a,ab),(b,ab), (aa,aab), (aa,abb),\\
& (ab,aab), (ab,abb), (ba,aab), (ba, abb)   \mid a,b \in A, a \neq b \}.
\end{align*}
Reasoning as in the sub-case $2.2)$ of Lemma \ref{lem:nonlengthmod},
\[
\binom{\varphi^{j-1}(a\varphi(u')b)}{ab^j} - \binom{\varphi^{j-1}(\varphi(v'))}{ab^j} = (m_{f^{j-1}(ab^j)}(ab)-m_{f^{j-1}(ab^j)}(ba)) (1+|u'| + |u'|_a - |v'|_a) >0.
\]
\end{itemize} 
\end{proof}

\begin{corollary}
\label{cor:smallWords}
Let $k \geq 3$. For all $n \leq 2^k-1$, we have $b_{\bt,k}(n) = p_\bt(n)$.
\end{corollary}
\begin{proof}
Let us take two different factors $u$ and $v$ of the same length $n \leq 2^k-1$. If $u$ and $v$ do not share any common prefix or suffix, $u \nsim_k v$ by the previous proposition. 
Otherwise, there exist words $x,y,u',v'$ such that $u=xu'y$, $v=xv'y$ where $u'$ and $v'$ do not share any common prefix or suffix. We apply the previous proposition to $u',v'$ and conclude using the cancellation property (Lemma~\ref{lem:cancel}).
\end{proof}

Let $u,v$ be distinct factors of $\bt$ of length $n \geq 2^k-1$. We are now ready to prove that $(p_u,s_u) \not \equiv_k (p_v,s_v)$ implies $u \nsim_k v$.

\begin{proof}[Proof of Theorem~\ref{thm:typediff}]

Let $x$ and $y$ respectively denote the longest common prefix and suffix of $u$ and $v$: $u = xu'y$ and $v = xv'y$.
We obviously have $u' \neq v'$ and, by Lemma \ref{lem:cancel}, we have $u \sim_k v$ if and only if $u' \sim_k v'$. 

If $|u'|\leq 2^k-1$, using Proposition~\ref{prop:smallWords}, we conclude that $u' \nsim_k v'$.
Otherwise, from Lemma~\ref{lem:deletePref}, $u'$ and $v'$ do not have the same type of order~$k$.
Thus, without loss of generality we may now assume that $u$ and $v$ do not have any non-empty common prefix or suffix.

Let $j$ be the greatest integer less than or equal to $k$ such that $|u|,|v| \geq 2^{j+1}$.
If $u$ and $v$ do not have the same type of order $j$, then we fall into one of the complementary situations of Lemmas~\ref{lem:nonlengthmod} or \ref{lem:lengthMod} (indeed, the extra assumption of Lemmas~\ref{lem:nonlengthmod} holds because $|u|\ge 2^{j+1}$, $|p_u|,|s_u|\le 2^j-1$ and thus  $(p_u,s_u),(p_v,s_v)$ satisfy $|p_u|+|s_u| < |u|$, $|p_v|+|s_v| < |v|$).
We thus have $u \nsim_j v$ and then $u \nsim_k v$.

Otherwise $u$ and $v$ have the same type of order $j$. By assumption, they do not have the same type of order~$k$, hence $j<k$. One has to do the same proof as the one of Proposition \ref{prop:smallWords}, except that one more argument is needed.
In case (1.a) and if $(u',v') \in \{(ab, ba), (aab,baa) \}$, we compute $\binom{u}{ab^{j+1}} - \binom{v}{ab^{j+1}}$. We need to stress the fact that in this particular case, $j < k-1$. Indeed, $j=k-1$ would give $u=\varphi^{k-1}(ab)$, $v= \varphi^{k-1}(ba)$ or $u=\varphi^{k-1}(a) \varphi^k(a)$, $v=\varphi^{k}(b) \varphi^{k-1}(a)$. In both cases, this is impossible since $u$ and $v$ do not have the same type of order~$k$.
\end{proof}

%===============================================================================================================

Due to Theorem~\ref{thm:typediff}, the $k$-binomial complexity of $\bt$ can be computed from 
\[
b_{\mathbf{t},k}(n)= \# \left( \Fac_n(\mathbf{t}) /\!\sim_k \right) =  \# \left( \{(p_u,s_u) : u \in \Fac_n(\mathbf{t}) \}/\!\equiv_k\right)\, .
\]

The last theorem provides this quantity.

\begin{theorem}
\label{thm:numberTypes}
For all $k \geq 3$, $n \geq 2^k$, we have
\begin{align*}
\# \left( \{(p_u,s_u) : u \in \Fac_n(\mathbf{t}) \}/\!\equiv_k\right) = \left\{\begin{array}{ll}
3 \cdot 2^k-3, & \text{if } n \equiv 0 \pmod{2^k} ;\\
3 \cdot 2^k-4, & \text{otherwise}.
\end{array}\right.
\end{align*}

\end{theorem}

\begin{proof}
Let $n \geq 2^k$ and set $\lambda \in \{0,\ldots,2^k-1 \}$ the integer such that $n \equiv \lambda \pmod{2^k}$.

For every $\ell \in \{0, \ldots, 2^{k-1}-1 \}$, 
\[
	P_\ell = \{ (p_u, s_u) : u \in \Fac_n(\mathbf{t}), |p_u|= \ell \, \text{ or } \, |p_u|=2^{k-1}+\ell \}
\]
and
\[
S_\ell = \{ (p_u, s_u) : u \in \Fac_n(\mathbf{t}), |s_u|= \ell \, \text{ or } \, |s_u|=2^{k-1}+\ell \}.
\]
If $\ell$ and $\ell'$ are two distinct elements from $\{0, \ldots, 2^{k-1}-1 \}$ then, due to Definition \ref{def:type}, for all $(p_u, s_u) \in P_\ell, (p_v,s_v) \in P_{\ell'}$, we have $(p_u,s_u) \not\equiv_k (p_v,s_v)$.
The idea of the proof is to count the number of equivalence classes in the pairwise disjoint sets $P_{\ell}$.

We can notice that $|p_u| + |s_u|\in \{\lambda, \lambda+2^k \}$ for all factorizations $(p_u,s_u)$.
From that, note that $P_0 = S_0$ if and only if $\lambda = 0$ or $\lambda=2^{k-1}$. In that case, we set $\ell_0=0$ and thus $P_{\ell_0}=S_0$. Otherwise, there exists $\ell_0 \neq 0$ such that $P_{\ell_0}=S_0$.

We will show that 
\begin{align*}
\# \left( (P_0 \cup P_{\ell_0}) /\!\equiv_k\right) =
\# \left( (P_0 \cup S_0) /\!\equiv_k\right) = \left\{\begin{array}{ll}
3, & \text{if } \lambda = 0;\\
2, & \text{if } \lambda =2^{k-1};\\
8, & \text{otherwise}
\end{array}\right.
\end{align*}
and that, for every $\ell \in \{0, \ldots, 2^{k-1}-1 \} \setminus \{ 0, \ell_0 \}$,
\begin{align*}
\# \left( P_\ell /\!\equiv_k\right)= 6.
\end{align*}

Obviously, we have
\[
\# \{0, \ldots, 2^{k-1}-1 \} \setminus \{ 0, \ell_0 \} = 
\left\{\begin{array}{ll}
2^{k-1}-1, & \text{if } \lambda = 0  \text{ or } \lambda=2^{k-1};\\
2^{k-1}-2, & \text{otherwise}.
\end{array}\right.
\]
Putting together the previous observations, we therefore get
\[
\# \left( \{(p_u,s_u) : u \in \Fac_n(\mathbf{t}) \}/\!\equiv_k\right) = \# \bigcup_{\ell=0}^{2^{k-1}-1} P_\ell = \left\{\begin{array}{ll}
6\,(2^{k-1}-1) +3, & \text{if } \lambda = 0;\\
6\,(2^{k-1}-1) +2, & \text{if } \lambda =2^{k-1};\\
6\,(2^{k-1}-2) +8, & \text{otherwise},
\end{array}\right.
\]
which gives us the expected result.

First, let us deal with $P_0$ and $S_0$. If $\lambda=0$, due to Proposition~\ref{pro:factorsappear} (ensuring that every pair appears and this argument is repeated all along the proof), we have $P_0 = S_0$ which is equal to
\[
\{ (\varepsilon, \varepsilon), (\varphi^{k-1}(0),\varphi^{k-1}(0)), (\varphi^{k-1}(0),\varphi^{k-1}(1)), (\varphi^{k-1}(1),\varphi^{k-1}(0)), (\varphi^{k-1}(1),\varphi^{k-1}(1)) \}.
\]
By Definition~\ref{def:type}, $\# \left( P_0 /\!\equiv_k \right)= 3$.
If $\lambda=2^{k-1}$, 
\[
P_0=S_0 = \{ (\varepsilon, \varphi^{k-1}(0)),(\varphi^{k-1}(0),\varepsilon), (\varepsilon, \varphi^{k-1}(1)),(\varphi^{k-1}(1),\varepsilon) \}
\]
and $\# \left( P_0 /\!\equiv_k \right)= 2$.
Finally, two sub-cases have to be distinguished if $\lambda \not\in \{0, 2^{k-1}\}$: either $0<\lambda < 2^{k-1}$ or $2^{k-1}<\lambda < 2^{k}$. Let $y$ be the prefix of $\varphi^{k}(0)$ of length $\lambda$. 

In the first sub-case, $y$ is also a prefix of $\varphi^{k-1}(0)$. We thus have
\begin{align*}
P_0 = \{ (\varepsilon,y), (\varepsilon, \overline{y}), &(\varphi^{k-1}(0), \varphi^{k-1}(1)y),(\varphi^{k-1}(1), \varphi^{k-1}(1)y), \\ &(\varphi^{k-1}(0), \varphi^{k-1}(0)\overline{y}),(\varphi^{k-1}(1), \varphi^{k-1}(0)\overline{y})
\}
\end{align*}
and $\# \left( P_0 /\!\equiv_k \right)=4$. We can proceed in the same way for $S_0$ and get a total of $8$ classes.

In the second sub-case, we can write $y = \varphi^{k-1}(0)z$ where $z$ is the prefix of $\varphi^{k-1}(1)$ of length $\lambda-2^{k-1}$.
We have
\begin{align*}
P_0=\{
(\varepsilon,\varphi^{k-1}(0)z), (\varepsilon,\varphi^{k-1}(1) \overline{z}), (\varphi^{k-1}(0),z),(\varphi^{k-1}(1),z),(\varphi^{k-1}(0),\overline{z}),(\varphi^{k-1}(1),\overline{z})
\}
\end{align*}
and once again, $\# \left( P_0 /\!\equiv_k \right)= 4$. The same result holds for $S_0$.
\\

Let us now consider $\ell \in \{0, \ldots, 2^{k-1}-1 \} \setminus \{ 0, \ell_0 \}$ and show that $\# \left( P_\ell /\!\equiv_k \right)=6$. Two cases have to be considered: either $\lambda< \ell$ or, $\lambda > \ell$. Indeed, we cannot have $\lambda=\ell$. Observe that if $\lambda=\ell$ or $\lambda=\ell + 2^{k-1}$, then $S_0=P_\ell$ which means that $\ell_0=\ell$ but we are assuming that $\ell\not\in\{ 0, \ell_0 \}$. Recall that that $|p_u| + |s_u|\in \{\lambda, \lambda+2^k \}$ for all factorizations $(p_u,s_u)$. We will make a constant use of this fact.
\begin{itemize}
  \item[a)] If $\lambda < \ell$, we cannot have $|p_u| + |s_u| = \lambda$, so obviously $|p_u| + |s_u| = 2^k + \lambda$ for all $(p_u,s_u) \in P_\ell$. Therefore, if $|p_u| = \ell < 2^{k-1}$, then $|s_u| > 2^{k-1}$. On the opposite, if $|p_u| = \ell + 2^{k-1}$, then $|s_u|<2^{k-1}$.
Set $x$ (resp., $y$) the suffix (resp., prefix) of $\varphi^{k-1}(0)$ of length $\ell$ (resp., $\lambda + 2^k - \ell$). 
We thus have 
\begin{align*}
P_\ell =\{
&(x, \varphi^{k-1}(0) \overline{y}), (x, \varphi^{k-1}(1)y),(\overline{x}, \varphi^{k-1}(0) \overline{y}),(\overline{x}, \varphi^{k-1}(1)y), \\ &(x \varphi^{k-1}(1),y),(x \varphi^{k-1}(1),\overline{y}),(\overline{x} \varphi^{k-1}(0),y),(\overline{x} \varphi^{k-1}(0),\overline{y})
\}
\end{align*}
and, from Definition~\ref{def:type}, $\# \left( P_\ell /\!\equiv_k \right)=  6$.

\item[b)] If $\ell < \lambda$, observe that $|p_u| = \ell \Rightarrow |s_u|=\lambda-\ell$. Indeed, since $|s_u|<2^k$, $|p_u| + |s_u|< \ell+2^k< \lambda+2^k$, hence we have $|p_u| + |s_u| = \lambda$. Two sub-cases have to be considered: $\lambda- \ell < 2^{k-1}$ or, $\lambda- \ell > 2^{k-1}$.
\begin{itemize}
  \item[b.1)] In the first sub-case, $|p_u| = \ell \Rightarrow |s_u| = \lambda-\ell < 2^{k-1}$. Otherwise stated, if $p_u$ is a suffix of some $\varphi^{k-1}(a)$, then $s_u$ is a prefix of some $\varphi^{k-1}(b)$.
Moreover, $\ell>\lambda-2^{k-1}$ ensures that $$|p_u| = \ell +2^{k-1} \Rightarrow |p_u|+|s_u| = \lambda+2^k.$$ Therefore, if $|p_u|=\ell+2^{k-1}$, then $|s_u| > 2^{k-1}$. Otherwise stated, if $p_u$ has a suffix of the form $\varphi^{k-1}(a)$, then $s_u$ has a prefix of the form $\varphi^{k-1}(b)$.
  \item[b.2)] In the second sub-case, $\ell+ 2^{k-1} < \lambda$ implies that
$$
|p_u|=\ell +2^{k-1} \Rightarrow |p_u| + |s_u| = \lambda.
$$
This is why, if $|p_u|=\ell +2^{k-1}$, then $|s_u| < 2^{k-1}$. Otherwise stated, if $p_u$ has a suffix of the form $\varphi^{k-1}(a)$, then $s_u$ is a prefix of some $\varphi^{k-1}(b)$. Finally, we already know that if $|p_u|=\ell$, then $|s_u|=\lambda-\ell$ which is here greater than $2^k-1$. Otherwise stated, if $p_u$ is a suffix of some $\varphi^{k-1}(a)$, then $s_u$ has a prefix of the form $\varphi^{k-1}(b)$.
\end{itemize}
Let us denote by $x$ the suffix of $\varphi^{k-1}(0)$ of length $\ell$ and $y$ the prefix of $\varphi^{k-1}(0)$, whose length is $\lambda - \ell$ in the first sub-case, $\lambda  - \ell- 2^{k-1}$ in the second one.
The case b.1) gives us 
\begin{align*}
P_\ell= \{
&(x, y), (\overline{x},y), (x, \overline{y}), (\overline{x}, \overline{y}), (x \varphi^{k-1}(1),  \varphi^{k-1}(1)y), \\&(x \varphi^{k-1}(1),  \varphi^{k-1}(0)\overline{y}), (\overline{x} \varphi^{k-1}(0),  \varphi^{k-1}(1)y), 
(\overline{x} \varphi^{k-1}(0),  \varphi^{k-1}(0)\overline{y})
\}
\end{align*}
while the case b.2) gives 
\begin{align*}
P_\ell = \{
&(x, \varphi^{k-1}(1)y), (\overline{x},\varphi^{k-1}(1)y), (x, \varphi^{k-1}(0)\overline{y}), (\overline{x}, \varphi^{k-1}(0)\overline{y}),\\ &(x \varphi^{k-1}(1),  y), (x \varphi^{k-1}(1),  \overline{y}), (\overline{x} \varphi^{k-1}(0),  y), 
(\overline{x} \varphi^{k-1}(0),  \overline{y})
\}.
\end{align*}
Both of them lead to the conclusion that $\# \left( P_\ell /\!\equiv_k \right)= 6$.
\end{itemize}

\end{proof}

As a consequence of Corollary~\ref{cor:smallWords}, Theorem~\ref{thm:typediff} and Theorem~\ref{thm:numberTypes}, we get the expected result stated in Theorem~\ref{the:main}.

\section{Acknowledgments}

We would like to thank Jeffrey Shallit for his participation in the statement of the initial problem. A conjecture about $b_{\mathbf{t},k}(n)$ was made when he was visiting the last author.


\begin{thebibliography}{99}
  \bibitem{AB07} B. Adamczewski, Y. Bugeaud, A short proof of the transcendence of Thue--Morse continued fractions, {\em Amer. Math. Monthly} {\bf 114} (2007), 536--540.
  \bibitem{AS99}  J.-P. Allouche, J. Shallit, The ubiquitous Prouhet-Thue-Morse sequence, {\em Sequences and their applications} (Singapore, 1998), 1--16, Springer Ser. Discrete Math. Theor. Comput. Sci., Springer, London, 1999.
  \bibitem{AS} J.-P. Allouche, J. Shallit, {\em Automatic sequences. Theory, applications, generalizations}, Cambridge University Press (2003).
  \bibitem{AF08} J.-P. Allouche, M. Mendès France, Euler, Pisot, Prouhet-Thue-Morse, Wallis and the duplication of sines, {\em Monatsh. Math.} {\bf 155} (2008), 301--315.
  \bibitem{Avg} S.V. Avgustinovich, D.G. Fon-Der-Flaass, A.E. Frid, Arithmetical complexity of infinite words, in: {\em Words, Languages and Combinatorics III} (Proc. 3rd ICWLC, Kyoto, March 2000), World Scientific, Singapore, 2003, pp. 51--62. 
  \bibitem{padic}  J. Berstel, M. Crochemore, J.-E. Pin, Thue-Morse sequence and $p$-adic topology for the free monoid, {\em Discrete Math.} {\bf 76} (1989),  89--94.
  \bibitem{CANT} V. Berth\'e, M. Rigo (Eds.), {\em Combinatorics,  Automata and Number Theory}, Encyclopedia Math. Appl., {\bf 135}, Cambridge Univ. Press (2010).
  \bibitem{Brlek} S. Brlek, Enumeration of factors in the Thue-Morse word, {\em Discrete Appl. Math.} {\bf 24} (1989), 83--96.
  \bibitem{CasFrid} J. Cassaigne, A. Frid, On the arithmetical complexity of Sturmian words, {\em Theoret. Comput. Sci.} {\bf 380} (2007), 304--316. 
  \bibitem{Dek77}  F. M. Dekking, Transcendance du nombre de Thue-Morse, {\em C. R. Acad. Sci. Paris} {\bf 285} (1977), 157--160.
  \bibitem{testing} D. D. Freydenberger, P. Gawrychowski, J. Karhumäki, F. Manea, W. Rytter, Testing $k$-binomial equivalence, arXiv:1509.00622.
  \bibitem{Grei} F. Greinecker, On the 2-abelian complexity of the Thue–Morse word, {\em Theoret. Comput. Sci.} {\bf 593} (2015), 88--105.
  \bibitem{Karandikar} P. Karandikar, M. Kufleitner, Ph. Schnoebelen, On the index of Simon's congruence for piecewise testability, {\em Inform. Process. Lett.} {\bf 115} (2015), 515--519.
  \bibitem{KS1} J. Karhum\"aki, A. Saarela, L. Q. Zamboni, On a generalization of Abelian equivalence and complexity of infinite words, {\em J. Combin. Theory Ser. A} {\bf 120} (2013), 2189--2206.
  \bibitem{KS2} J. Karhum\"aki, A. Saarela, L. Q. Zamboni, Variations of the Morse-Hedlund theorem for $k$-abelian equivalence, {\em Lect. Notes in Comput. Sci.} {\bf 8633} (2014), 203--214.
  \bibitem{LM}  J. Lambek, L. Moser, On some two way classifications of integers, {\em Canad. Math. Bull.} {\bf 2} (1959), 85--89.
  \bibitem{Lejeune} M. Lejeune, {\em Au sujet de la complexité $k$-binomiale}, Master thesis, University of Li\`ege (2018), {\tt http://hdl.handle.net/2268.2/5007}.
  \bibitem{LRS}  J. Leroy, M. Rigo, M. Stipulanti, Generalized Pascal triangle for binomial coefficients of words, {\em Adv. in Appl. Math.} {\bf 80} (2016), 24--47.
  \bibitem{Lot} M. Lothaire, {\em Combinatorics on words}, Cambridge Mathematical Library. Cambridge University Press (1997).
  \bibitem{mos1} B. Moss\'e, Puissances de mots et reconnaissabilit\'e des points fixes d'une substitution, {\em Theoret. Comput. Sci.} {\bf 99} (1992), 327--334.
  \bibitem{mos2} B. Moss\'e, Reconnaissabilit\'e des substitutions et complexit\'e des suites automatiques, {\em Bull. Soc. Math. France} {\bf 124} (1996), 329--346. 
  \bibitem{walnut} H. Mousavi, Automatic Theorem Proving in Walnut, arXiv:1603.06017.
  \bibitem{Och} P. Ochsenschl\"ager, Binomialkoeffizienten und Shuffle-Zahlen, Technischer Bericht, Fachbereicht Informatik, T.H. Darmstadt (1981).
\bibitem{Pin} J.-\'E. Pin, P. V. Silva, A noncommutative extension of Mahler's theorem on interpolation series, {\em European J. Combin.} {\bf 36} (2014), 564--578.
  \bibitem{Pytheas}  N. Pytheas Fogg, {\em Substitutions in Dynamics, Arithmetics and Combinatorics}, Lect. Notes in Math. {\bf 1794}, V. Berth\'e et al. Eds, Springer (2002).
  \bibitem{Rao} M. Rao, M. Rigo, P. Salimov, Avoiding $2$-binomial squares and cubes, {\em Theoret. Comput. Sci.} {\bf 572} (2015), 83--91.
  \bibitem{Rigobook} M. Rigo, {\em Formal Languages, Automata and Numeration Systems 1, Introduction to Combinatorics on Words}, Network and Telecommunications series, ISTE-Wiley (2014).
  \bibitem{RigoSalimov:2015} M. Rigo, P. Salimov, Another generalization of abelian equivalence: binomial complexity of infinite words, {\em Theoret. Comput. Sci.} {\bf 601} (2015), 47--57.
  \bibitem{survey} M. Rigo, Relations on words, {\em Indag. Math. (N.S.)} {\bf 28} (2017), 183--204. 
\end{thebibliography}
\end{document}